\documentclass[11pt,letterpaper]{article}

\usepackage[utf8]{inputenc} 

\usepackage{geometry} 
\usepackage{graphicx} 

\usepackage[usenames,dvipsnames,svgnames,table]{xcolor}
\usepackage[colorlinks, citecolor=black, urlcolor=black, linkcolor=black]{hyperref}

\usepackage{amsmath}
\usepackage{amsthm}
\usepackage{amsfonts}
\usepackage{mathrsfs}
\usepackage{listings}
\usepackage{tabto}
\usepackage{array} 
\usepackage{verbatim}
\usepackage{nicefrac}

\usepackage{enumitem}
\usepackage{booktabs} 
\usepackage{tabularx}

\usepackage{algorithm}
\usepackage[noend]{algpseudocode}

\usepackage{tikz}
\usetikzlibrary{patterns}
\usetikzlibrary{decorations,arrows}
\usetikzlibrary{decorations.pathmorphing}
\usepgflibrary{decorations.pathreplacing} 
\usetikzlibrary{positioning}
\definecolor{verylightgray}{rgb}{0.001,0.001,0.001}

\usepackage{caption}
\usepackage{subcaption}
\providecommand{\keywords}[1]{\textbf{keywords.} #1}

\newtheoremstyle{thm}
	{11pt}       
	{11pt}      
	{\itshape}  
	{}          
	{\bfseries} 
	{:}         
	{ }      
	{}          

\theoremstyle{thm}

\newtheorem{lma}{Lemma}
\newtheorem{thm}{Theorem}

\newtheoremstyle{definition}
	{11pt}       
	{0pt}      
	{\normalfont}  
	{}          
	{\bfseries} 
	{:}         
	{ }      
	{}          
\theoremstyle{definition}

\providecommand{\bpm}{\begin{pmatrix}}
\providecommand{\epm}{\end{pmatrix}}

\newcommand{\eps}{\varepsilon}

\newcommand{\OPT}{\mathrm{OPT}}

\newcommand{\tallItemHeight}{(1/3 +\eps)\OPT}
\newcommand{\boxL}{\mathcal{B}_L}
\newcommand{\boxT}{\mathcal{B}_T}
\newcommand{\boxH}{\mathcal{B}_H}
\newcommand{\boxV}{\mathcal{B}_V}
\newcommand{\boxS}{\mathcal{B}_{S}}
\newcommand{\boxSV}{\mathcal{B}_{S,V}}
\newcommand{\boxSH}{\mathcal{B}_{S,H}}
\newcommand{\boxTV}{\mathcal{B}_{T\cup V}}
\newcommand{\hatboxTV}{\hat{\mathcal{B}}_{T\cup V}}
\newcommand{\checkboxTV}{\check{\mathcal{B}}_{T\cup V}}

\title{Improved approximation for two dimensional strip packing with polynomial bounded width\thanks{Research was supported in part by German Research Foundation (DFG) project JA 612 /14-2. An extended abstract of this paper was published at WALCOM 2017 \cite{JansenR17}}}
\author{Klaus Jansen, Malin Rau\\
Institute of Computer Science, University of Kiel, 24118 Kiel, Germany\\
\{kj,mra\}@informatik.uni-kiel.de}
\date{}

\begin{document}
\maketitle
\begin{abstract}
We study the well-known two-dimensional strip packing problem. Given a set of rectangular axis-parallel items and a strip of width $W$ with infinite height, the objective is to find a packing of all items into the strip, which minimizes the packing height. Lately, it has been shown that the lower bound of $3/2$ of the absolute approximation ratio can be beaten when we allow a pseudo-polynomial running-time of type $(n W)^{f(1/\eps)}$. If $W$ is polynomially bounded by the number of items, this is a polynomial running-time.
The currently best pseudo-polynomial approximation algorithm by Nadiradze and Wiese achieves an approximation ratio of $1.4+\eps$. 
We present a pseudo-polynomial algorithm with improved approximation ratio $4/3 +\eps$. 
Furthermore, the presented algorithm has a significantly smaller running-time as the $1.4+\eps$ approximation algorithm.
\end{abstract}
\section{Introduction}
An instance of the strip packing problem consists of a strip of width $W\in\mathbb{N}$ and infinite height and a set of items $I$, where each item $i\in I$ has width $w_i\in\mathbb{N}$ and height $h_i\in\mathbb{N}$, such that all items fit into the strip (i.e $w_i \leq W$ f.a. $i \in I$). 
 
A packing of the items is a mapping $\rho: I \rightarrow \mathbb{N}\times\mathbb{N}, i \mapsto (x_i,y_i)$, where $x_i \leq W - w_i$. We say an \textit{inner point} of a placed item $i$ is a point $(x,y) \in \mathbb{N}\times \mathbb{N}$, with $y_i \leq y < y_i +h_i$ and $x_i \leq x < x_i +w_i$.
We say two items $i$ and $j$ \textit{overlap} if there exists a point $(x,y) \in \mathbb{N}\times \mathbb{N}$, such that $(x,y)$ is an inner point of $i$ and an inner point of $j$. A packing is \textit{feasible} if no two items overlap. The objective is to find a feasible packing, which minimizes its height $\max_{i \in I} y_i + h_i$. For a set of items $S$ we denote its area by $A(S) := \sum_{i \in S} h_iw_i$. We denote the packing area by $W \times \max_{i \in I} y_i + h_i$.

Strip packing is one of the classical two-dimensional packing problems, which received a high research interest \cite{AdamaszekKPP16, Baker2, Baker, Coffman, Golan, harren20145, HenningJRS17, jansenSolisOba, JansenThoele, Kenyon00, nadiradzeWiese, Schiermeyer, Sleator, Steinberg,Khan}. It arises naturally in many practical applications as manufacturing and logistics as well as in computer science. There are many manufacturing settings where rectangular pieces have to be cut out of some sheet of raw material. If the packing height is minimized, the unused area which can be seen as the waste of the raw material is minimized as well. In computer science strip packing can be used to model scheduling parallel jobs on consecutive addresses. Here the width $W$ of the strip equals the number of given processors.

If $W$ occurs polynomially in the running time, it is called pseudo-polynomial. If $W \leq poly(n)$ the running time can be considered polynomial.
The algorithm with the so far best absolute approximation ratio using pseudo-polynomial running time is the algorithm by Nadiradze and Wiese \cite{nadiradzeWiese}. Their algorithm has an absolute approximation ratio of $1.4 +\eps$.

\paragraph{Results and methodology.} 
Let $\OPT$ be the height of an optimal packing.
We present an algorithm with absolute approximation ratio $4/3 +\eps$, which has a pseudo-polynomial running time. 
The main difficulty arises when placing items which have a small width and a large height. 
If the considered algorithm can not place all these items into an optimal packing's area, it would have to place it above this area, adding its height to the height of the packing. 
Since these items can have a height up to $\OPT$ this can double the height of the packing.

In \cite{nadiradzeWiese} Nadiradze and Wiese presented a new technique to handle tall items, which have small width and height larger than $0.4 \OPT$. 
They managed to place all these items into an optimal packing's area. 
In this packing some of the items with height up to $0.4 \OPT$ are shifted upwards and are placed above this area. 
These shifted items are responsible for adding $0.4 \OPT$ to the absolute approximation ratio. 

We present a stronger structural result than in \cite{nadiradzeWiese}, leading to an algorithm, that can place all items with height at least $\frac{1}{3}\OPT$ in an optimal packing's area. By this optimization just items with height up to $\frac{1}{3} \OPT$ have to be placed above this area, which results in an approximation algorithm with absolute approximation ratio $4/3 +\eps$. This is possible since we could reduce the area of the items with height smaller than $\frac{1}{3} \OPT$ that have to be shifted on top of the optimal packing area. The key to this better approximation lies in Lemma \ref{lma:verticalItemShift}.

\begin{figure} [ht]
\centering
\resizebox{.95\textwidth}{!}{%
\begin{tikzpicture}
\draw[very thick] (0,0) rectangle (10,11);
\draw [dotted](-0.5,11) -- (10.5,11) node[right] {\huge $(1 +\mathcal{O}(\eps))\OPT$};
\draw[very thick]  (0,11) rectangle (10,11.5) node[midway]{\Large medium sized items};
\draw [very thick] (0,11.5) rectangle (10,12.5) node[midway]{\Large remaining horizontal items};
\draw[very thick]  (0,12.5) rectangle (10,13) node[midway]{\Large remaining small items};
\draw[very thick]  (0,13) rectangle (0.5,13+22/5) node[midway,rotate=90]{\Large medium items};
\draw[very thick]  (0.5,13) rectangle (1.5,13+22/5) node[midway,rotate=90]{\Large vertical items};
\draw[very thick]  (5,13) rectangle (10,13+11/5)node[midway]{\Large shifted vertical items};
\draw[very thick]  (5,13+11/5) rectangle (10,13+22/5)node[midway]{\Large shifted vertical items};
\draw[very thick, fill=gray] (0,0) rectangle (3,3);
\draw[very thick, fill=gray] (5,4) rectangle (7,8);
\draw[very thick, fill=gray] (0,7) rectangle (4,5);
\draw[very thick, fill=gray] (0,7) rectangle (3,11);
\draw[very thick, fill=gray] (8,0) rectangle (10,9);
\draw[fill=gray] (6,0.8) rectangle (8,1);
\draw[fill=gray] (3,0.8) rectangle (6,1);
\draw[fill=gray] (3,0) rectangle (7.5,0.4);
\draw[fill=gray] (5,0.4) rectangle (7,0.8);
\draw[fill=gray] (3,0.4) rectangle (5,0.8);
\draw[very thick]  (3,0) rectangle (8,1);
\draw[fill=gray] (3,3.8) rectangle (6,4);
\draw[fill=gray] (0,3.8) rectangle (3,4);
\draw[fill=gray] (0,3) rectangle (7,3.4);
\draw[fill=gray] (5,3.4) rectangle (6.5,3.7);
\draw[fill=gray] (2,3.4) rectangle (5,3.8);
\draw[fill=gray] (0,3.4) rectangle (2,3.8);
\draw[very thick]  (0,3) rectangle (7,4);
\draw[fill=gray] (3,4.7) rectangle (5,4.9);
\draw[fill=gray] (0,4.8) rectangle (3,5);
\draw[fill=gray] (0,4) rectangle (5,4.4);
\draw[fill=gray] (3,4.4) rectangle (5,4.7);
\draw[fill=gray] (2,4.4) rectangle (3,4.8);
\draw[fill=gray] (0,4.4) rectangle (2,4.8);
\draw[very thick]  (0,4) rectangle (5,5);
\draw[fill=gray] (5,10) rectangle (7,10.2);
\draw[fill=gray] (5,10.2) rectangle (6.8,10.3);
\draw[fill=gray] (5,10.3) rectangle (6.8,10.5);
\draw[fill=gray] (5,10.5) rectangle (6.6,10.6);
\draw[fill=gray] (7,10) rectangle (9.7,10.2);
\draw[fill=gray] (7,10.2) rectangle (9.6,10.5);
\draw[fill=gray] (7,10.5) rectangle (9.4,10.6);
\draw[fill=gray] (5,10.6) rectangle (8.9,10.8);
\draw[fill=gray] (5,10.8) rectangle (8.6,10.9);
\draw[very thick]  (5,10) rectangle (10,11);
\draw[fill=gray]  (5,9) rectangle (7.5,9.2);
\draw[fill=gray]  (5,9.2) rectangle (7.3,9.3);
\draw[fill=gray]  (5,9.3) rectangle (7.2,9.5);
\draw[fill=gray]  (7.5,9) rectangle (9.7,9.2);
\draw[fill=gray]  (7.5,9.2) rectangle (9.5,9.3);
\draw[fill=gray]  (7.5,9.3) rectangle (9.4,9.6);
\draw[fill=gray]  (5,9.6) rectangle (8.7,9.8);
\draw[fill=gray]  (5,9.8) rectangle (8.6,9.9);
\draw[very thick]  (5,9) rectangle (10,10);
\draw[fill=gray]  (5,8) rectangle (6.9,8.1);
\draw[fill=gray]  (5,8.1) rectangle (6.6,8.4);
\draw[fill=gray]  (5,8.4) rectangle (6.4,8.7);
\draw[fill=gray]  (5,8.7) rectangle (6.1,8.8);
\draw[very thick]  (5,8) rectangle (7,9);
\draw[fill=gray]  (7,1) rectangle (7.2,9);
\draw[fill=gray]  (7.2,1) rectangle (7.4,9);
\draw[very thick]  (7,1) rectangle (7.5,9);
\draw[fill=gray]  (7.5,1) rectangle (7.7,6);
\draw[fill=gray]  (7.7,1) rectangle (8,6);
\draw[very thick]  (7.5,1) rectangle (8,6);
\draw[fill=gray]  (7.5,6) rectangle (7.7,9);
\draw[fill=gray]  (7.7,6) rectangle (7.9,9);
\draw[very thick]  (7.5,6) rectangle (8,9);
\draw[fill=gray]  (4,5) rectangle (4.3,7);
\draw[fill=gray]  (4.3,5) rectangle (4.4,7);
\draw[fill=gray]  (4.4,5) rectangle (4.6,7);
\draw[fill=gray]  (4.6,5) rectangle (4.8,7);
\draw[very thick]  (4,5) rectangle (5,7);
\draw[fill=gray]  (4,7) rectangle (4.2,8.5);
\draw[fill=gray]  (4.2,7) rectangle (4.6,8.5);
\draw[fill=gray]  (4.6,7) rectangle (4.9,8.5);
\draw[very thick]  (4,7) rectangle (5,8.5);
\draw[fill=gray]  (4,8.5) rectangle (4.3,11);
\draw[fill=gray]  (4.3,8.5) rectangle (4.5,11);
\draw[fill=gray]  (4.5,8.5) rectangle (5,11);
\draw[very thick]  (4,8.5) rectangle (5,11);
\draw[fill=gray]  (3,9) rectangle (3.3,11);
\draw[fill=gray]  (3.3,9) rectangle (3.5,11);
\draw[fill=gray]  (3.5,9) rectangle (3.8,11);
\draw[very thick]  (3,9) rectangle (4,11);
\draw[fill=gray]  (3,7) rectangle (3.3,9);
\draw[fill=gray]  (3.3,7) rectangle (3.6,9);
\draw[fill=gray]  (3.6,7) rectangle (3.9,9);
\draw[very thick]  (3,7) rectangle (4,9);
\draw[fill=gray]  (3,1) rectangle (3.4,3);
\draw[fill=gray]  (3.4,1) rectangle (3.5,3);
\draw[fill=gray]  (3.5,1) rectangle (3.9,3);
\draw[fill=gray]  (3.9,1) rectangle (4.2,3);
\draw[fill=gray]  (4.2,1) rectangle (4.7,3);
\draw[fill=gray]  (4.7,1) rectangle (5,3);
\draw[fill=gray]  (5,1) rectangle (5.1,3);
\draw[fill=gray]  (5.1,1) rectangle (5.3,3);
\draw[fill=gray]  (5.3,1) rectangle (5.4,3);
\draw[fill=gray]  (5.4,1) rectangle (5.8,3);
\draw[fill=gray]  (5.8,1) rectangle (5.9,3);
\draw[fill=gray]  (5.9,1) rectangle (6.3,3);
\draw[fill=gray]  (6.3,1) rectangle (6.6,3);
\draw[very thick]  (3,1) rectangle (7,3);
\draw [thick,decorate,decoration={brace,amplitude=12pt}] (10,13) -- (10,11) node [midway,right,xshift=12pt](B){\huge $\mathcal{O}(\epsilon)\OPT$};
\draw [thick,decorate,decoration={brace,amplitude=12pt}] (10,13+22/5) -- (10,13) node [midway,right,xshift=12pt](B){\huge $(\frac{2}{5} +\mathcal{O}(\epsilon))\OPT$};
\draw [thick,decorate,decoration={brace,amplitude=12pt}] (0,13+22/5) -- (1.5,13+22/5) node [midway,above,yshift=11pt](B){\huge $\mathcal{O}(\epsilon) W$};
\draw [thick,decorate,decoration={brace,amplitude=12pt}] (5,13+22/5) -- (10,13+22/5) node [midway,above,yshift=11pt](B){\huge $\frac{1}{2} W$};
\end{tikzpicture}
\begin{tikzpicture}

\draw[very thick] (0,0) rectangle (10,11);
\draw [dotted](-0.5,10) -- (10.5,10) node[right] {\huge $\OPT$};
\draw [dotted](-0.5,11) -- (10.5,11) node[right] {\huge $(1+5\eps)\OPT$};
\draw[very thick]  (0,11) rectangle (10,11.5) node[midway]{\Large{medium sized item}s};
\draw [very thick] (0,11.5) rectangle (10,12.5) node[midway]{\LARGE{some horizontal items}};
\draw[very thick]  (0,12.5) rectangle (1,12.5 +11/3) node[midway,rotate=90]{\Large{medium items}};
\draw[very thick]  (1,12.5) rectangle (10,12.5 +11/3)node[midway]{\Large{shifted vertical item}s};

\draw[very thick, fill=lightgray] (0,0) rectangle (3,3);
\draw[very thick, fill=lightgray] (5,4) rectangle (7,8);
\draw[very thick, fill=lightgray] (0,7) rectangle (4,5);
\draw[very thick, fill=lightgray] (0,7) rectangle (3,11);
\draw[very thick, fill=lightgray] (8,0) rectangle (10,9);

\draw[fill=lightgray] (6,0.8) rectangle (8,1);
\draw[fill=lightgray] (3,0.8) rectangle (6,1);
\draw[fill=lightgray] (3,0) rectangle (7.5,0.4);
\draw[fill=lightgray] (5,0.4) rectangle (7,0.8);
\draw[fill=lightgray] (3,0.4) rectangle (5,0.8);

\draw[pattern=north east lines]  (3,0) rectangle (8,1);

\draw[fill=lightgray] (3,3.8) rectangle (6,4);
\draw[fill=lightgray] (0,3.8) rectangle (3,4);
\draw[fill=lightgray] (0,3) rectangle (7,3.4);
\draw[fill=lightgray] (5,3.4) rectangle (6.5,3.7);
\draw[fill=lightgray] (2,3.4) rectangle (5,3.8);
\draw[fill=lightgray] (0,3.4) rectangle (2,3.8);

\draw[pattern=north east lines]  (0,3) rectangle (7,4);

\draw[fill=lightgray] (3,4.7) rectangle (5,4.9);
\draw[fill=lightgray] (0,4.8) rectangle (3,5);
\draw[fill=lightgray] (0,4) rectangle (5,4.4);
\draw[fill=lightgray] (3,4.4) rectangle (5,4.7);
\draw[fill=lightgray] (2,4.4) rectangle (3,4.8);
\draw[fill=lightgray] (0,4.4) rectangle (2,4.8);

\draw[pattern=north east lines]  (0,4) rectangle (5,5);

\draw[fill=lightgray] (5,10) rectangle (7,10.2);
\draw[fill=lightgray] (5,10.2) rectangle (6.8,10.3);
\draw[fill=lightgray] (5,10.3) rectangle (6.8,10.5);
\draw[fill=lightgray] (5,10.5) rectangle (6.6,10.6);
\draw[fill=lightgray] (7,10) rectangle (9.7,10.2);
\draw[fill=lightgray] (7,10.2) rectangle (9.6,10.5);
\draw[fill=lightgray] (7,10.5) rectangle (9.4,10.6);
\draw[fill=lightgray] (5,10.6) rectangle (8.9,10.8);
\draw[fill=lightgray] (5,10.8) rectangle (8.6,10.9);

\draw[pattern=north east lines]  (5,10) rectangle (10,11);

\draw[fill=lightgray]  (5,9) rectangle (7.5,9.2);
\draw[fill=lightgray]  (5,9.2) rectangle (7.3,9.3);
\draw[fill=lightgray]  (5,9.3) rectangle (7.2,9.5);
\draw[fill=lightgray]  (7.5,9) rectangle (9.7,9.2);
\draw[fill=lightgray]  (7.5,9.2) rectangle (9.5,9.3);
\draw[fill=lightgray]  (7.5,9.3) rectangle (9.4,9.6);
\draw[fill=lightgray]  (5,9.6) rectangle (8.7,9.8);
\draw[fill=lightgray]  (5,9.8) rectangle (8.6,9.9);
\draw[pattern=dots]  (5,9) rectangle (10,10);

\draw[fill=lightgray]  (5,8) rectangle (6.9,8.1);
\draw[fill=lightgray]  (5,8.1) rectangle (6.6,8.4);
\draw[fill=lightgray]  (5,8.4) rectangle (6.4,8.7);
\draw[fill=lightgray]  (5,8.7) rectangle (6.1,8.8);
\draw[pattern=dots]  (5,8) rectangle (7,9);

\draw[fill=lightgray]  (7,1) rectangle (7.2,9);
\draw[fill=lightgray]  (7.2,1) rectangle (7.4,9);
\draw[pattern=dots]  (7,1) rectangle (7.5,9);

\draw[fill=lightgray]  (7.5,1) rectangle (7.7,6);
\draw[fill=lightgray]  (7.7,1) rectangle (8,6);
\draw[pattern=dots]  (7.5,1) rectangle (8,6);

\draw[fill=lightgray]  (7.5,6) rectangle (7.7,9);
\draw[fill=lightgray]  (7.7,6) rectangle (7.9,9);
\draw[pattern=dots]  (7.5,6) rectangle (8,9);

\draw[fill=lightgray]  (4,5) rectangle (4.3,7);
\draw[fill=lightgray]  (4.3,5) rectangle (4.4,7);
\draw[fill=lightgray]  (4.4,5) rectangle (4.6,7);
\draw[fill=lightgray]  (4.6,5) rectangle (4.8,7);
\draw[pattern=dots]  (4,5) rectangle (5,7);

\draw[fill=lightgray]  (4,7) rectangle (4.2,8.5);
\draw[fill=lightgray]  (4.2,7) rectangle (4.6,8.5);
\draw[fill=lightgray]  (4.6,7) rectangle (4.9,8.5);
\draw[pattern=dots]  (4,7) rectangle (5,8.5);

\draw[fill=lightgray]  (4,8.5) rectangle (4.3,11);
\draw[fill=lightgray]  (4.3,8.5) rectangle (4.5,11);
\draw[fill=lightgray]  (4.5,8.5) rectangle (5,11);
\draw[pattern=dots]  (4,8.5) rectangle (5,11);

\draw[fill=lightgray]  (3,9) rectangle (3.3,11);
\draw[fill=lightgray]  (3.3,9) rectangle (3.5,11);
\draw[fill=lightgray]  (3.5,9) rectangle (3.8,11);
\draw[pattern=dots]  (3,9) rectangle (4,11);

\draw[fill=lightgray]  (3,7) rectangle (3.3,9);
\draw[fill=lightgray]  (3.3,7) rectangle (3.6,9);
\draw[fill=lightgray]  (3.6,7) rectangle (3.9,9);
\draw[pattern=dots]  (3,7) rectangle (4,9);

\draw[fill=lightgray]  (3,1) rectangle (3.4,3);
\draw[fill=lightgray]  (3.4,1) rectangle (3.5,3);
\draw[fill=lightgray]  (3.5,1) rectangle (3.9,3);
\draw[fill=lightgray]  (3.9,1) rectangle (4.2,3);
\draw[fill=lightgray]  (4.2,1) rectangle (4.7,3);
\draw[fill=lightgray]  (4.7,1) rectangle (5,3);
\draw[fill=lightgray]  (5,1) rectangle (5.1,3);
\draw[fill=lightgray]  (5.1,1) rectangle (5.3,3);
\draw[fill=lightgray]  (5.3,1) rectangle (5.4,3);
\draw[fill=lightgray]  (5.4,1) rectangle (5.8,3);
\draw[fill=lightgray]  (5.8,1) rectangle (5.9,3);
\draw[fill=lightgray]  (5.9,1) rectangle (6.3,3);
\draw[fill=lightgray]  (6.3,1) rectangle (6.6,3);
\draw[pattern=dots]  (3,1) rectangle (7,3);

\draw[very thick, red]  (3,0) rectangle (8,1);
\draw[very thick, red]  (0,3) rectangle (7,4);
\draw[very thick, red]  (0,4) rectangle (5,5);
\draw[very thick, red]  (5,10) rectangle (10,11);
\draw[very thick, red]  (5,9) rectangle (10,10);
\draw[very thick, red] (5,8) rectangle (7,9);
\draw[very thick, red]  (7,1) rectangle (8,9);

\draw[very thick, red]  (4,5) rectangle (5,11);
\draw[very thick, red]  (3,7) rectangle (4,11);
\draw[very thick, red]  (3,1) rectangle (7,3);

\draw[very thick, red] (0,0) rectangle (3,3);
\draw[very thick, red] (5,4) rectangle (7,8);
\draw[very thick, red] (0,7) rectangle (4,5);
\draw[very thick, red] (0,7) rectangle (3,11);
\draw[very thick, red] (8,0) rectangle (10,9);

\draw [thick,decorate,decoration={brace,amplitude=12pt}] (10,12.5) -- (10,11) node [midway,right,xshift=12pt](B){\huge $3\eps\OPT'$};

\draw [thick,decorate,decoration={brace,amplitude=12pt}] (10,12.5 +11/3) -- (10,12.5) node [midway,right,xshift=12pt](B){\huge $\tallItemHeight$};

\draw [thick,decorate,decoration={brace,amplitude=6pt}] (0,12.5 +11/3) -- (1,12.5 +11/3) node [midway,above,yshift=11pt](B){\huge $(3\eps/2)W$};

\draw [thick,decorate,decoration={brace,amplitude=12pt}] (1,12.5 +11/3) -- (10,12.5 +11/3) node [midway,above,yshift=11pt](B){\huge $(1-3\eps/2)W$};
\end{tikzpicture}

}
\caption{Comparison of the structural results. Left the structural result, which leads to $7/5 +\eps$ and right the new structural result, which leads to $4/3+\eps$.
}
\label{fig:partition}
\end{figure}
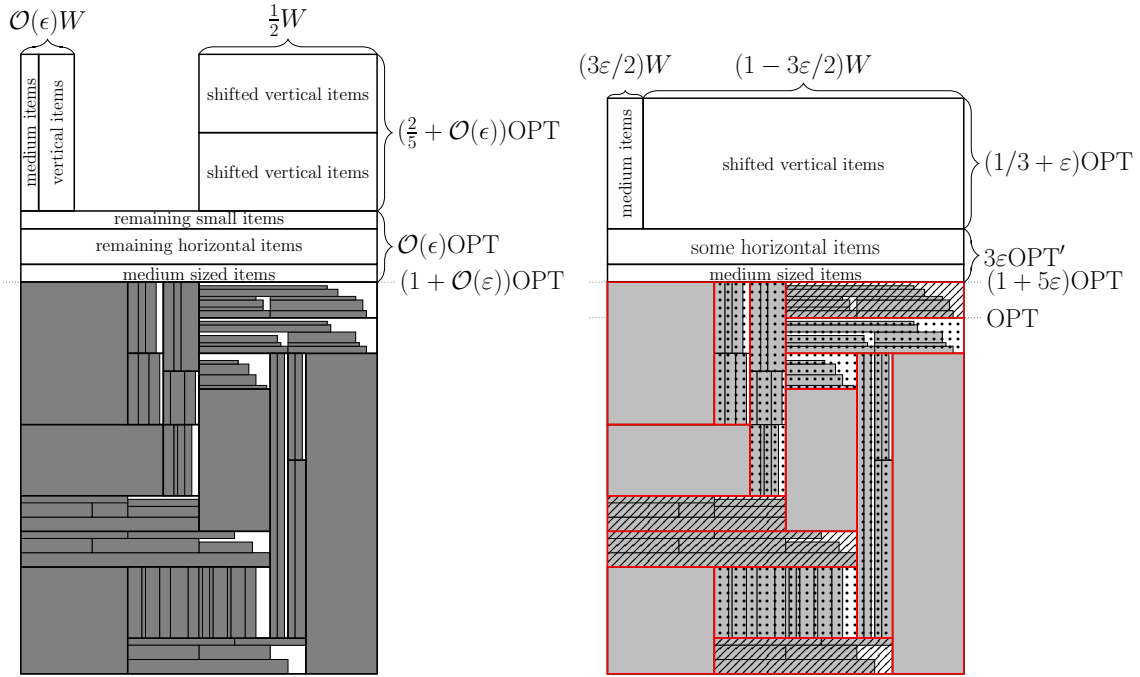

The second improvement to the algorithm in \cite{nadiradzeWiese} lies in the running time of the algorithm. The main idea in \cite{nadiradzeWiese} is to divide the packing area into a constant number of rectangular areas. The number of these areas depends on $\eps$ and can be quite large (i.e. $\Omega(6^{1/\delta})$). Since the width of each of these areas has to be guessed and the number of these boxes influences the choice of $\delta$ this induces a very large running time, i.e. $\mathcal{O}(W^{1/\delta})$, where in the worst case $\delta \in \Omega(1/\exp_6^{1/\eps}(1/\eps))$, where $\exp_6^{1/\eps}(1/\eps) := 6^{\dots^{6^{1/\eps}}}$ and the $6$ occurs $1/\eps$ times (tower of exponents).
We manage to reduce the number of these areas dramatically (i.e. $\mathcal{O}(1/\eps^3\delta^2)$ which  implies $\delta \in \Omega(\eps^{\mathcal{O}(2^{1/\eps})})$). How we find this better partition is described in the proof of Lemma \ref{lma:partitionOflargeBoxes}. So the result of our research is summarized in the following Theorem:

\begin{thm}
For each $\eps > 0$ there is an algorithm that finds a solution for each instance of the strip packing problem with height at most $(4/3 +\eps)\OPT$. The algorithm needs at most $(nW)^{1/\eps^{\mathcal{O}(2^{1/\eps})}}$ operations.
\end{thm}

An algorithm with the same approximation ratio was developed independently and at the same time by G\'alvez, Grandoni, Ingala and Khan \cite{Khan}. They extended their approach to strip packing with rotations, but did not improve the running time.

\paragraph{Related work.}
The first algorithm for the strip packing problem was described by Baker and Coffman \cite{Baker} in 1980. If the rectangles are ordered by descending width, this algorithm has an asymptotic approximation ratio of 3.
The first algorithms with proven absolute approximation ratios of $3$ and $2.7$ were given by Coffman, Garey, Jonson and Tarjan \cite{Coffman}.
After that Sleator \cite{Sleator} presented an algorithm which generates a schedule of height $2OPT(I) + h_{\max}(I)/2$, where $h_{\max}$ is the largest height of the items. So this algorithm has an asymptotic approximation ratio $2$. Schiermeyer \cite{Schiermeyer} and Steinberg \cite{Steinberg} improved this algorithm independently to an algorithm with absolute approximation ratio $2$. Harren and van Stee were the first to beat the barrier of $2$. They presented an algorithm with an absolute approximation ratio of $1.9396$. The so far best absolute approximation is given by the algorithm by Harren, Jansen, Pr\"adel and van Stee \cite{harren20145}, which has an absolute approximation ratio of $(5/3 + \epsilon)OPT(I)$. A reduction from the partition problem gives a lower bound on the absolute approximation ratio of $3/2 \cdot \OPT$ for any polynomial approximation algorithm.  

In the asymptotic case, the barrier of $3/2$ can be beaten. 
Golan \cite{Golan} presented the first algorithm with asymptotic approximation ratio smaller than $3/2$. It has an asymptotic approximation ratio of $4/3$. Next Baker \cite{Baker2} gave an algorithm with asymptotic ratio $5/4$. After that Kenyon and R\'emila \cite{Kenyon00} presented an AFPTAS which has an approximation ratio of $(1+\epsilon)\OPT$ and an additive constant $\mathcal{O}(h_{\max}/\epsilon^2)$. Later the additive constant was improved by Jansen and Solis-Oba \cite{jansenSolisOba} at the expense of the processing time of the algorithm. They presented an APTAS, which generates a schedule of height $(1+\epsilon)\OPT + h_{\max}$.

If we allow pseudo-polynomial processing time, there are better approximations possible.   This is thanks to the fact that the underlying partition problem is solvable in pseudo-polynomial time. Jansen and Th\"ole \cite{JansenThoele} presented an algorithm with approximation ratio $3/2 +\epsilon$. Recently Nadiradze and Wiese \cite{nadiradzeWiese} have presented an algorithm which beats the bound $3/2$. It has an approximation ratio of $1.4 +\epsilon$. On the negative side Adamaszek et al. \cite{AdamaszekKPP16} have shown that there is no pseudo-polynomial algorithm with approximation ratio smaller than $\frac{12}{11} \OPT$. This could be improved to a lower bound of $\frac{5}{4}\OPT$ in \cite{HenningJRS17}.

\paragraph{Organization of this paper.}
In the following sections we will prove a structural result, which leads to the algorithm with approximation ratio $(4/3 +\eps)\OPT$. Given an optimal packing with hight $\OPT$, we describe how it can be transformed into an other packing with a certain structure. Since each optimal packing can be transformed, the algorithm simply needs to guess the structure and fill the items via dynamic programming into it. 
In section \ref{Simplify} we describe adjustments to simplify the given set of items. We use these simplifications to find the structure as well as to speed up the packing algorithm.
In section \ref{ratio} we describe the key to find the improved approximation.  In section \ref{runningTime} we describe how the running time can be improved, by reducing number of different possible structures of the transformed packing. In section \ref{Algorithm} we describe the algorithm which finds a packing with height at most $(4/3 +\eps)\OPT$.

\section{Simplifying the input instance}
\label{Simplify}
Let $\eps > 0$, such that $1/\eps \in \mathbb{N}$. Further, let an instance of the strip packing problem be given and consider an optimal solution to it, which has a packing height of $\OPT$. 
Notice that we can find the height of the optimal packing by a binary search framework in $\mathcal{O}(\log(\OPT))$ steps, which is polynomial in the input size. The described algorithm would also work if we would approximate the optimal packing height within the range of $(1 + \mathcal{O}(\eps))$, which would result in $\mathcal{O}(\log(1/\eps))$ steps of the framework. However, for the simplification of the notation we use the exact height.

The first step in the transformation as well as in the algorithm is to partition the set of items $I$. Let $\delta = \delta(\eps) >\mu = \mu(\eps)$ be two suitable constants depending on $\eps$. We define the set of 
large items $L := \{i \in I| h_i \geq \delta \OPT, w_i \geq \delta W\}$,
 tall items $T := \{i \in I \setminus L | h_i \geq \tallItemHeight\}$,
 vertical items $V := \{i \in I \setminus T|h_i \geq \delta \OPT, w_i \leq \mu W\}$,
 medium sized vertical items $M_V := \{i \in I \setminus T|h_i \geq \delta \OPT, \mu W < w_i < \delta W\}$,
horizontal items $H := \{i \in I|h_i \leq \mu \OPT, w_i \geq \delta W\}$,
 small items $\{i \in I| h_i \leq \mu \OPT, w_i \leq \mu W\}$ and
 medium sized horizontal items $M_H := I \setminus (L \cup T \cup V \cup M_V \cup H \cup S)$.

As usual, the medium sized items will be placed outside the optimal packing area. To guarantee that these items do not use to much space outside the packing area, we have to ensure that the total area of these items is small. We achieve this by finding appropriate values for $\delta$ and $\mu$. In the following Lemma, we show that such values do exist. It is a standard argument which follows by the pigeon-hole principle and is often used in packing algorithms, e.g. in \cite{jansenSolisOba}.
 
\begin{lma}
\label{lma:sizeDelta}
Consider the sequence $\sigma_0 = x\eps^{y}$, $\sigma_{i+1} = \sigma_{i}^{z}\eps^y$. There is a value $j \in \{0, \dots, f(1/\eps)-1\}$ such that when defining $\delta = \sigma_j$  and $\mu = \sigma_{j+1}$ the total area of the items in $M_V \cup M_H$ is at most $f(\eps)\cdot \OPT \cdot W$.
\end{lma}

\begin{proof}
This follows by the pigeon-hole principle: The sequence delivers a partition of the set of items into $1/f(\eps)$ disjunctive sets. If the area of each of this sets is larger than $f(\eps) \cdot \OPT \cdot W$, then their total area is larger than $\OPT \cdot W$, which is a contradiction, since all items fit into the area $W \times \OPT$.
\end{proof}

For the structural result it is sufficient to define $x := 1$, $y := 6$, $z := 2$ and $f(\eps):= \eps/6$.
Note that $\sigma_i = \eps^{6(2^{i+1}-1)}$. Since $\sigma$ is strictly monotonic decreasing we have $\delta \geq \sigma_{6/\eps-1}$. So we have the following lower bound:
 $\delta \geq \eps^k$, for $k= 6\cdot 2^{6/\eps}$. 
Since the area of the medium sized items is small, they can be placed above the packing without using too much extra space. 

The next step in our transformation is to round the heights of the items in $L \cup T \cup V$ and shift them such that they start and end at certain heights. Our rounding strategy is similar to the strategy in \cite{nadiradzeWiese} but we manage to reduce the number of different heights. The next Lemma describes our rounding procedure more formally.

\begin{lma}
\label{lma:rounding}
Let $\delta = \eps^k$ for some value $k \in \mathbb{N}$. At a loss of at most a factor $1+2\eps$ in the approximation ratio we can ensure that each item $i \in L \cup T \cup V$ with $\eps^{l-1}  \OPT > h_i \geq \eps^{l}  \OPT$ for some $l\in\mathbb{N}_{\leq k}$ has height $h_i' = k_i  \eps^{l+1}  \OPT$ for some $k_i \in \{1/\eps,\dots,1/\eps^2\}$. Furthermore the items' y-coordinates can be placed at multiples of $\eps^{l+1} \OPT$. 
\end{lma}

\begin{proof}
Since the rounding strategy is similar to \cite{nadiradzeWiese}  the proof is with exception of the choice of $\gamma$ quite analogue. 

Let a packing in the strip of height $\OPT$ be given. We stretch it by a factor of $1+2\eps$. 
This means each point $(x,y)$ in the original strip corresponds to the point $(x,(1+2\eps)y)$ in the stretched packing. 
Let $i \in L\cup T \cup V$ be an item with $\eps^{l-1} \cdot \OPT \geq h_i \geq \eps^{l} \cdot \OPT$ and 
let $y_T$ and $y_B$ be the y-coordinates of its top and bottom edges, respectively, in the original strip. 
Furthermore, we define the stretched y-coordinates as $\bar{y}_T := (1 +2\eps)y_T$ and $\bar{y}_B := (1 +2\eps)y_B$. 
As a consequence we have $\bar{y}_T - \bar{y}_B = (1+2\eps)(y_t -y_b) = (1+2\eps)h_i$. 
Now we change the y-coordinates of $i$ in the new strip to $y_B' := \bar{y}_B + \eps h_i$ and $y_T' := \bar{y}_T - \eps h_i$. 
We get that $y_T' - y_B' = \bar{y}_T - \bar{y}_b - 2 \eps h_i = (1+2\eps)h_i - 2\eps h_i = h_i$. 
We have $h_i \geq \eps^l \OPT$, which implies $\bar{y}_T -y'_T = y'_B -\bar{y}_B = \eps h_i \geq \eps^{l+1} \cdot \OPT$. 
This ensures, that for the interval $[y'_T, \bar{y}_T]$ there is an integer $k_T$ such that $k_T \cdot\eps^{l+1} \cdot \OPT\in [y'_T, \bar{y}_T]$, analogously there exists an integer $k_B$ such that $k_B \cdot\eps^{l+1} \cdot \OPT \in [\bar{y}_B, y'_B]$. 
We change the y-coordinates of item $i$ to $y''_T = k_T \cdot \eps^{l+1} \cdot \OPT$ and $y''_B = k_B \cdot \eps^{l+1} \cdot \OPT$. 
It can happen, that $y''_T -y''_B > h_i' := \lceil h_i / (\eps^{l+1} \OPT) \rceil \cdot \eps^{l+1} \cdot \OPT$. 
In this case we increase $y''_B$ by $(y''_T -y''_B) - h'_i$, such that $h_i' = y''_T -y''_B$. 
Item $i$ does not intersect an other item, since it is placed inside of the stretched version of itself.  
Thus when we change the height $h_i$ of each item $i \in L \cup T \cup V$ to $h'_i = \lceil h_i / (\eps^{l+1} \OPT) \rceil \cdot \eps^{l+1} \cdot \OPT $, where $l \in \mathbb{N}$ is chosen such that $\eps^{l-1}\OPT\geq h_i \geq \eps^l\OPT$.
Note that $\lceil h_i / (\eps^{l+1} \OPT) \rceil \in \{1/\eps,\dots,1/\eps^2\}$ since $\eps^{l-1}\OPT\geq h_i \geq \eps^l\OPT$.
Since the $h'$ does not exceed the stretched item height, we increased the optimal solution value by at most a factor $1+2\eps$.
\end{proof}

How many different heights do we get by this rounding strategy? Each item $i$ with height $h_i \in [\eps^{l-1},\eps^l]$ we round  to the next larger multiple of $\eps^{l+1} \cdot \OPT$. Since $1/\eps^2 \cdot \eps^{l+1} = \eps^{l-1}$ and $1/\eps \cdot \eps^{l+1} = \eps^l$ we have at most $1/\eps^2 -1/\eps \leq 1/\eps^2$ different multiples of $\eps^{l+1}$ in the interval $[\eps^{l-1},\eps^l]$. Since $\delta \geq \eps^k$ we have at most $k$ of this intervals. So in total we have at most $k/\eps^2$ different sizes.

\section{Improving the approximation ratio}
\label{ratio}
We apply the rounding according to Lemma \ref{lma:rounding}, obtaining a packing where each item in $L \cup T \cup V$ starts and ends at multiples of $\delta\eps$. The rounded packing has a height of $(1+2\eps)\OPT$. Similar to \cite{nadiradzeWiese}, we will show that we can partition its packing area into a constant number of rectangular areas, such that each of these areas contains items just from one of the following sets: $L$, $H \cup S$, or $T \cup V \cup S$. We allow items from $H \cup S$ or $T \cup V \cup S$ to be positioned into more than one area. We will see that there are simple algorithms to place the items from $L$ or $H$ into their rectangular areas, while it is still difficult to place the items from $T \cup V$, without increasing the height of the packing too much. Note that there are at most $1/\delta^2$ large items since they cover an area of at least $\delta^2 W \OPT $.

\begin{lma}
\label{lma:roughPartition}
We can partition the area $W \times (1 +2\eps)\OPT$ into at most $ 4(1+2\eps)/(\eps\delta^2)$ rectangular areas called boxes. The set of these boxes can be partitioned into sets $\boxL, \boxH$ and $\boxTV$ such that
\begin{itemize}
\item boxes in $\boxL$ are identified by items $i\in L$, i.e. they have box height $h_i$ and box width $w_i$,
\item $\mathcal{B}_H$ consists of at most $(1+2\eps)/(\eps\delta^2) - |L|/\delta$ many boxes of height $\eps\delta \OPT$, each of them containing at least some item in $H$ but only items in $H\cup S$,
\item $\boxTV$ consists of at most $3(1+2\eps)/(\eps\delta^2)$ many boxes, each of them containing items in $T\cup V\cup S$,
\item no item in $H$ is intersected vertically by any box border,
\item no item in $T \cup V$ is intersected horizontally by any box border.
\end{itemize}
\end{lma}

\begin{proof}
Let us consider our stretched optimal packing in the strip with height $(1+2\eps)\OPT$. 
In the first step, we give each item in $L$ its personal box, which has exactly the dimensions of that item. 
Since the item does not overlap any other item the box does not either. 

In the next step, we define boxes for the horizontal items: We iterate over the strips of height $\eps\delta$ from bottom to top. We know that each item in $L\cup T \cup V$ starts and ends at a multiple of $\eps\delta$, so none of this items does start within one of these strips. 

We partition the strips in the following way: we start at the left of the strip and iterate to the right until we meet the first item in the set $H$ or in the set $L\cup T \cup V$. We remember which item we met first, and draw a vertical line when we meet the first item out of the other set. 
We remember from which set the item came we had just met. We iterate further to the right until we again met an item from the other set or the border of the strip.  
If we have not yet met the border of the strip, we draw a vertical line and continue as before. When we met the border of the strip, we look at our vertical lines. 
If we look at the set of items between two vertical lines we see that they either contain items from $L \cup T \cup V$ or items from $H$, but there is no set which contains items from $H$ and $L \cup T \cup V$ as well. 
The area between two vertical lines, which contains items from $H$ defines a box for items in $H$. 
Since each item in $H$ has a width of at least $\delta W$ we get at most $1/\delta -1$ of these boxes for horizontal items per horizontal strip of height $\eps\delta$. Since we have $(1+2\eps)/(\eps\delta)$ of these strips we get at most $(1+2\eps)/(\eps\delta) \cdot (1/\delta -1) \leq (1+2\eps)/(\eps\delta^2) -2$ of these boxes. We call the set of this boxes $\mathcal{B}_H$.

Now we describe how to get the boxes for the items in $T \cup V$: For each of the boxes $\mathcal{B}_H$ and the items in $L$ we draw vertical lines on the left and on the right side, until they meet the first item out of $L$ or the first box in $\mathcal{B}_H$. 
The area, which is bounded within two of this lines, defines a box for vertical and tall items.
We call this set of boxes $\mathcal{B}_{T\cup V}$.
All together we have at most $(1+2\eps)/(\eps\delta^2) -2$ large items and boxes in $\mathcal{B}_H$. Each of this boxes or items produces two lines. Additionally, we have the strip border which gives two additional lines. Each line touches at most $3$ boxes in $\mathcal{B}_{V \cup T}$. Each of this boxes needs two lines as a border. So in total we have at most  $3(1+2\eps)/(\eps\delta^2)$ boxes in $\mathcal{B}_{V\cup T}$.
\end{proof}

Let us considerer boxes $B \in \boxTV$ with $h(B)$ at least $(2/3+2\eps)\OPT$. The key in \cite{nadiradzeWiese} was to rearrange the items in this boxes, such that the tall items can be placed into a constant number of subboxes, which contain just items with the same rounded height. By this rearrangement many vertical items have to be shifted above the optimal packing area. The key for a better approximation is to show the possibility that some of these items can be placed back into this rearranged packing. We will prove this possibility in Lemma \ref{lma:verticalItemShift}.

For simplification, we remove all small items from the boxes $\boxH$ and $\boxTV$. Let $\hatboxTV \subseteq \boxTV$ be the set of boxes with height at least $(2/3 +2\eps)\OPT$ and $\checkboxTV := \boxTV \setminus \hatboxTV$. Let us assume that we are allowed to slice all vertical items horizontally as often as we desire. If we consider a packing of items, where some of the vertical items could be sliced vertically, we call it fractional packing. We call all tall items, which are not crossed by any box border movable items and all other tall items unmovable items.
The first step in the rearrangement is to shift tall items up or down respectively such that all movable tall items either touch the top or the bottom of the box.The existence of this rearrangement was already shown (see Lemma 1.4 in \cite{nadiradzeWiese}).

\begin{lma}[\cite{nadiradzeWiese}]
\label{lma:shifttallItems}
If we are allowed to slice the items in $V$ vertically, we can ensure the following:
In each box $B \in \boxTV$ there is a packing where all movable tall items are either touching the top or the bottom of the box.
\end{lma}

Let us from now on assume that all movable tall items are touching the top or the bottom of the boxes in $\boxTV$. In Section \ref{runningTime}, we will reorder the tall items, such that we generate few subboxes for tall items.
It can happen that not all vertical items can be placed into the box after this reordering. All vertical items that can not be placed have to be shifted above the packing area. Since we have just the area $W \times (1/3+\eps)\OPT$ to pack the shifted items, we have to be careful, not to shift too many items. 
For this purpose, we introduce pseudo items, which only contain vertical items and touch the bottom or the top of a box in $\boxTV$, as described in \cite{nadiradzeWiese}.

For $B \in \hatboxTV$ let $(x_l,y_b)$ be the left bottom corner and $(x_r,y_t)$ the top right corner respectively. Let $X = \{x_1,x_2, \dots, x_{k-1}\}$ be the x-coordinates of the tall items in the packing, ordered in increasing order and define $x_0 = x_l$ and $x_k = x_r$. Consider a pair $x_{j-1},x_j$. 
If $[x_{j-1},x_j) \times [y_b,y_t)$ does not overlap any tall item, we introduce one pseudo item with size $[x_{j-1},x_j) \times [y_b,y_t)$. 
Consider the case that $[x_{j-1},x_j) \times [y_b,y_t)$ overlaps with exactly one tall item $i$ of height $h_i$. 
If $i$ is touching the bottom we introduce one pseudo item which covers the area $[x_{j-1},x_j) \times [y_b +h_i,y_t)$ and if $i$ touches the top boundary we introduce a pseudo item which covers the area $[x_{j-1},x_j) \times [y_b,y_t -h_i)$. 
The last case is that $[x_{j-1},x_j) \times [y_b,y_t)$ overlaps exactly two tall items. In this case we introduce no pseudo item. Let $P$ be the set of all the introduced pseudo items. All vertical items that are crossed by a pseudo item border are sliced along that border. Note that vertical items cross only vertical pseudo item borders. Let $B_V$ be the set of all (slices of) vertical items, which are contained in $B$ but not covered by any pseudo item. 

If we reorder the tall and pseudo items in a box $B \in \hatboxTV$, it can happen, that it is not possible to place all the items in $B_V$ in $B$. Unlike in \cite{nadiradzeWiese} we have to ensure that at least a constant amount of these items can be placed in $B$.
For each x-coordinate $x_i \in [x_l,\dots,x_r-1]\cap \mathbb{N}$ let $b \in T\cup P$ be the item, which touches the bottom of $B$ and $t \in T\cup P$ be the item touching the top, each intersecting the x-axis at $x_i + 1/2$. We define a container $C_i$ which touches $t$ and $b$ and spans from $x_i$ to $x_i+1$. Let $\mathcal{C}_B$ be the set of all container for a given placement of tall and pseudo items in $B$ (see figure \ref{fig:ContainerForverticalItems}). 
A reordering of the tall and pseudo $T \cup P$ items in $B$ is a rearrangement, which just changes the x-coordinates of the bottom-left corners, but not the y-coordinates. It is feasible if there are no two items in $T \cup P$ that overlap in this reordering.

Next we will show that in any reordering of the items, there is a constant amount of containers, which can be placed into the box without overlapping with any other container or (pseudo) item.

\begin{lma}
\label{lma:verticalItemShift}
Let $\beta := \min \{|h_i - h_j| : i,j \in T\cup P, h_i \not = h_j \}$ be the minimal difference between the heights of two items in $T \cup P$. Let $B \in \hatboxTV$ with width $w$. Let $t \in T\cup P$ be the shortest item touching the top and $b \in T\cup P$ be the shortest item touching the bottom, with $h_t >0$ and $h_b >0$. Define $h := (h(B) - h_t - h_b)$.
For each feasible reordering of the tall and pseudo items and each $\alpha \leq \beta/(\beta +h)$, we can find a subset $S \subseteq \mathcal{C}_B$ of the containers for $B_V$, with $|S| \geq \alpha w$ that can be placed in the reordered packing.
\end{lma}

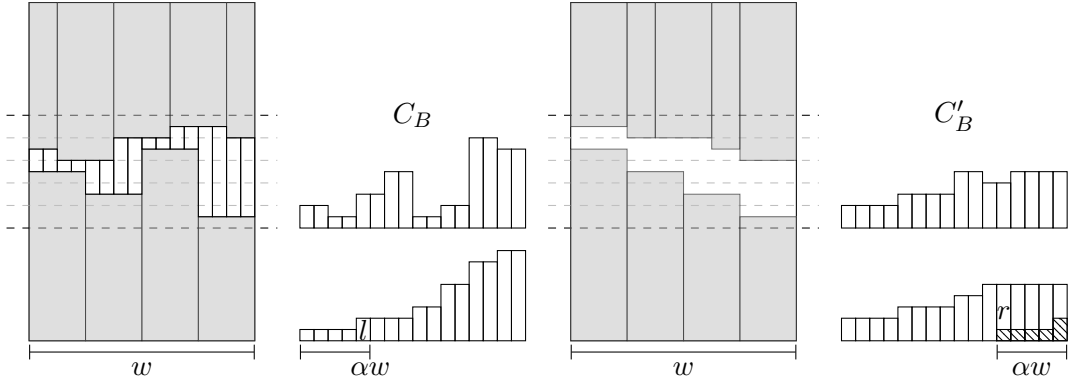
\begin{figure}[h]{}
\centering
\begin{tikzpicture}

\pgfmathsetmacro{\h}{1.5}
\pgfmathsetmacro{\w}{3}

\pgfmathsetmacro{\wa}{0.0625}
\pgfmathsetmacro{\wb}{0.125}
\pgfmathsetmacro{\wc}{0.1875}

\pgfmathsetmacro{\wz}{0.25}
\pgfmathsetmacro{\we}{0.3125}
\pgfmathsetmacro{\wf}{0.375}
\pgfmathsetmacro{\wg}{0.4375}
\pgfmathsetmacro{\wh}{0.5}
\pgfmathsetmacro{\wi}{0.5625}
\pgfmathsetmacro{\wj}{0.625}
\pgfmathsetmacro{\wk}{0.6875}
\pgfmathsetmacro{\wl}{0.75}
\pgfmathsetmacro{\wm}{0.8125}
\pgfmathsetmacro{\wn}{0.875}
\pgfmathsetmacro{\wo}{0.9375}

\draw (0,0) rectangle (\w,3*\h);

\draw[|-|] (0,-0.1*\h) -- (\w,-0.1*\h) node[midway,below] {$w$};

\draw[dashed] (-0.1*\w,\h) -- (1.1*\w,\h);
\draw[dashed] (-0.1*\w,2*\h) -- (1.1*\w,2*\h);

\draw[lightgray, dashed] (-0.1*\w,1.2*\h) -- (1.1*\w,1.2*\h);
\draw[lightgray, dashed] (-0.1*\w,1.4*\h) -- (1.1*\w,1.4*\h);
\draw[lightgray, dashed] (-0.1*\w,1.6*\h) -- (1.1*\w,1.6*\h);
\draw[lightgray, dashed] (-0.1*\w,1.8*\h) -- (1.1*\w,1.8*\h);

\draw[fill= lightgray, opacity = 0.5] (0,0) rectangle (0.25 *\w,1.5*\h);
\draw[fill= lightgray, opacity = 0.5] (0.25 *\w,0) rectangle (0.5 *\w,1.3*\h);
\draw[fill= lightgray, opacity = 0.5] (0.5 *\w,0) rectangle (0.75 *\w,1.7 *\h);
\draw[fill= lightgray, opacity = 0.5] (0.75 *\w,0) rectangle (\w,1.1*\h);

\draw[fill= lightgray, opacity = 0.5] (0,1.7*\h) rectangle (0.125 *\w,3*\h);
\draw[fill= lightgray, opacity = 0.5] (0.125 *\w,1.6*\h) rectangle (0.375 *\w,3*\h);
\draw[fill= lightgray, opacity = 0.5] (0.375 *\w,1.8*\h) rectangle (0.625 *\w,3*\h);
\draw[fill= lightgray, opacity = 0.5] (0.625 *\w,1.9*\h) rectangle (0.875 *\w,3*\h);
\draw[fill= lightgray, opacity = 0.5] (0.875 *\w,1.8*\h) rectangle (\w,3*\h);

\draw (0,1.5*\h) rectangle (\wa *\w,1.7*\h);
\draw (\wa *\w,1.5*\h) rectangle (\wb *\w,1.7*\h);
\draw (\wb *\w,1.5*\h) rectangle (\wc *\w,1.6*\h);
\draw (\wc *\w,1.5*\h) rectangle (\wz *\w,1.6*\h);
\draw (\wz *\w,1.3*\h) rectangle (\we *\w,1.6*\h);
\draw (\we *\w,1.3*\h) rectangle (\wf *\w,1.6*\h);
\draw (\wf *\w,1.3*\h) rectangle (\wg *\w,1.8*\h);
\draw (\wg *\w,1.3*\h) rectangle (\wh *\w,1.8*\h);
\draw (\wh *\w,1.7 *\h) rectangle (\wi *\w,1.8*\h);
\draw (\wi *\w,1.7 *\h) rectangle (\wj *\w,1.8*\h);
\draw (\wj *\w,1.7 *\h) rectangle (\wk *\w,1.9*\h);
\draw (\wk *\w,1.7 *\h) rectangle (\wl *\w,1.9*\h);
\draw (\wl *\w,1.1 *\h) rectangle (\wm *\w,1.9*\h);
\draw (\wm *\w,1.1 *\h) rectangle (\wn *\w,1.9*\h);
\draw (\wn *\w,1.1 *\h) rectangle (\wo * \w,1.8*\h);
\draw (\wo *\w,1.1 *\h) rectangle (\w,1.8*\h);

\begin{scope}[xshift= 1.2*\w cm]

\node at (0.5*\w, 2*\h){$C_B$};
\begin{scope}[yshift= 1*\h cm]
\draw (0,0) rectangle (\wa *\w,0.2*\h);
\draw (\wa *\w,0) rectangle (\wb *\w,0.2*\h);
\draw (\wb *\w,0) rectangle (\wc *\w,0.1*\h);
\draw (\wc *\w,0) rectangle (\wz*\w,0.1*\h);
\draw (\wz *\w,0) rectangle (\we *\w,0.3*\h);
\draw (\we *\w,0) rectangle (\wf *\w,0.3*\h);
\draw (\wf *\w,0) rectangle (\wg *\w,0.5*\h);
\draw (\wg *\w,0) rectangle (\wh *\w,0.5*\h);
\draw (\wh *\w,0) rectangle (\wi *\w,0.1*\h);
\draw (\wi *\w,0) rectangle (\wj *\w,0.1*\h);
\draw (\wj *\w,0) rectangle (\wk *\w,0.2*\h);
\draw (\wk *\w,0) rectangle (\wl *\w,0.2*\h);
\draw (\wl *\w,0) rectangle (\wm *\w,0.8*\h);
\draw (\wm *\w,0) rectangle (\wn *\w,0.8*\h);
\draw (\wn *\w,0) rectangle (\wo*\w,0.7*\h);
\draw (\wo *\w,0) rectangle (\w,0.7*\h);

\end{scope}

\begin{scope}[yshift= 0*\h cm]
\draw (0,0) rectangle (\wa *\w,0.1*\h);
\draw (\wa *\w,0) rectangle (\wb *\w,0.1*\h);
\draw (\wb *\w,0) rectangle (\wc *\w,0.1*\h);
\draw (\wc *\w,0) rectangle (\wz *\w,0.1*\h);
\draw (\wz *\w,0) rectangle (\we *\w,0.2*\h) node[midway]{$l$};
\draw (\we *\w,0) rectangle (\wf *\w,0.2*\h);
\draw (\wf *\w,0) rectangle (\wg *\w,0.2*\h);
\draw (\wg *\w,0) rectangle (\wh *\w,0.2*\h);
\draw (\wh *\w,0) rectangle (\wi *\w,0.3*\h);
\draw (\wi *\w,0) rectangle (\wj *\w,0.3*\h);
\draw (\wj *\w,0) rectangle (\wk *\w,0.5*\h);
\draw (\wk *\w,0) rectangle (\wl *\w,0.5*\h);
\draw (\wl *\w,0) rectangle (\wm *\w,0.7*\h);
\draw (\wm *\w,0) rectangle (\wn *\w,0.7*\h);
\draw (\wn *\w,0) rectangle (\wo*\w,0.8*\h);
\draw (\wo *\w,0) rectangle (\w,0.8*\h);

\draw[|-|] (0,-0.1*\h) -- (\we*\w,-0.1*\h) node[below] {$\alpha w$};
\end{scope}
\end{scope}

\begin{scope}[xshift=2.4*\w cm]
\draw (0,0) rectangle (\w,3*\h);

\draw[|-|] (0,-0.1*\h) -- (\w,-0.1*\h) node[midway,below] {$w$};

\draw[dashed] (-0.1*\w,\h) -- (1.1*\w,\h);
\draw[dashed] (-0.1*\w,2*\h) -- (1.1*\w,2*\h);

\draw[lightgray, dashed] (-0.1*\w,1.2*\h) -- (1.1*\w,1.2*\h);
\draw[lightgray, dashed] (-0.1*\w,1.4*\h) -- (1.1*\w,1.4*\h);
\draw[lightgray, dashed] (-0.1*\w,1.6*\h) -- (1.1*\w,1.6*\h);
\draw[lightgray, dashed] (-0.1*\w,1.8*\h) -- (1.1*\w,1.8*\h);

\draw[fill= lightgray, opacity = 0.5] (0,0) rectangle (0.25 *\w,1.7*\h);
\draw[fill= lightgray, opacity = 0.5] (0.25 *\w,0) rectangle (0.5 *\w,1.5*\h);
\draw[fill= lightgray, opacity = 0.5] (0.5 *\w,0) rectangle (0.75 *\w,1.3 *\h);
\draw[fill= lightgray, opacity = 0.5] (0.75 *\w,0) rectangle (\w,1.1*\h);

\draw[fill= lightgray, opacity = 0.5] (0,1.9*\h) rectangle (0.25 *\w,3*\h);
\draw[fill= lightgray, opacity = 0.5] (0.25 *\w,1.8*\h) rectangle (0.375 *\w,3*\h);
\draw[fill= lightgray, opacity = 0.5] (0.375 *\w,1.8*\h) rectangle (0.625 *\w,3*\h);
\draw[fill= lightgray, opacity = 0.5] (0.625 *\w,1.7*\h) rectangle (0.75 *\w,3*\h);
\draw[fill= lightgray, opacity = 0.5] (0.75 *\w,1.6*\h) rectangle (\w,3*\h);

\begin{scope}[xshift= 1.2*\w cm]

\node at (0.5*\w, 2*\h){$C_B'$};

\begin{scope}[yshift= 1*\h cm]
\draw (0,0) rectangle (\wa *\w,0.2*\h);
\draw (\wa *\w,0) rectangle (\wb *\w,0.2*\h);
\draw (\wb *\w,0) rectangle (\wc *\w,0.2*\h);
\draw (\wc *\w,0) rectangle (\wz *\w,0.2*\h);
\draw (\wz *\w,0) rectangle (\we *\w,0.3*\h);
\draw (\we *\w,0) rectangle (\wf *\w,0.3*\h);
\draw (\wf *\w,0) rectangle (\wg *\w,0.3*\h);
\draw (\wg *\w,0) rectangle (\wh *\w,0.3*\h);
\draw (\wh *\w,0) rectangle (\wi *\w,0.5*\h);
\draw (\wi *\w,0) rectangle (\wj *\w,0.5*\h);
\draw (\wj *\w,0) rectangle (\wk *\w,0.4*\h);
\draw (\wk *\w,0) rectangle (\wl *\w,0.4*\h);
\draw (\wl *\w,0) rectangle (\wm*\w,0.5*\h);
\draw (\wm *\w,0) rectangle (\wn *\w,0.5*\h);
\draw (\wn *\w,0) rectangle (\wo *\w,0.5*\h);
\draw (\wo *\w,0) rectangle (\w,0.5*\h);
\end{scope}

\begin{scope}[yshift= 0*\h cm]
\draw (0,0) rectangle (\wa*\w,0.2*\h);
\draw (\wa *\w,0) rectangle (\wb *\w,0.2*\h);
\draw (\wb *\w,0) rectangle (\wc *\w,0.2*\h);
\draw (\wc *\w,0) rectangle (\wz *\w,0.2*\h);
\draw (\wz *\w,0) rectangle (\we *\w,0.3*\h);
\draw (\we *\w,0) rectangle (\wf *\w,0.3*\h);
\draw (\wf *\w,0) rectangle (\wg *\w,0.3*\h);
\draw (\wg *\w,0) rectangle (\wh *\w,0.3*\h);
\draw (\wh *\w,0) rectangle (\wi *\w,0.4*\h);
\draw (\wi *\w,0) rectangle (\wj *\w,0.4*\h);
\draw (\wj *\w,0) rectangle (\wk *\w,0.5*\h);
\draw (\wk *\w,0) rectangle (\wl *\w,0.5*\h)node[midway]{$r$};
\draw (\wl *\w,0) rectangle (\wm * \w,0.5*\h);
\draw (\wm *\w,0) rectangle (\wn *\w,0.5*\h);
\draw (\wn *\w,0) rectangle (\wo *\w,0.5*\h);
\draw (\wo *\w,0) rectangle (\w,0.5*\h);

\draw[pattern= north west lines] (\wk *\w,0) rectangle (\wl *\w,0.1*\h);
\draw[pattern= north west lines] (\wl *\w,0) rectangle (\wm *\w,0.1*\h);
\draw[pattern= north west lines] (\wm *\w,0) rectangle (\wn *\w,0.1*\h);
\draw[pattern= north west lines] (\wn *\w,0) rectangle (\wo *\w,0.1*\h);
\draw[pattern= north west lines] (\wo *\w,0) rectangle (\w,0.2*\h);

\draw[|-|] (\wk * \w,-0.1*\h) -- (\w,-0.1*\h) node[midway,below] {$\alpha w$};
\end{scope}
\end{scope}

\end{scope}

\end{tikzpicture}
\caption{Two orderings of the items in $T \cup P$.}
\label{fig:ContainerForverticalItems}
\end{figure}

\begin{proof}
Let $C_B$ be the set of containers for vertical items in the first ordering and $C_B'$ be the set in a given feasible reordering. We sort both sets of containers in ascending order and index them from $1$ to $w$. We will show that the $\lceil \alpha w \rceil$ smallest containers in $C_B$ fit into the $\lceil \alpha w \rceil$ largest containers in $C_B'$.
Let $l$ be the container with index $\lceil \alpha w \rceil$ in the set $C_B$ and let $r$ be the container with index $w -\lceil \alpha w \rceil +1$ in the set $C_B'$. If $h_l \leq h_r$ the $\lceil \alpha w \rceil$ shortest container in $\mathcal{C}_B$ can be placed into the $\lceil \alpha w \rceil$ longest container $\mathcal{C}_B'$, see Figure \ref{fig:ContainerForverticalItems}.

Assume for contradiction that $h_l > h_r$. We know about the area of the sets of containers that $A(C_B) = A(C_B')$, since we have not changed the set of tall and pseudo items. Since each container with index $\geq l$ has height at least $h_l$, we know that $A(C_B) \geq h_l(w-\lceil \alpha w \rceil +1)$. Furthermore we know that $A(C_B')\leq h_r(w-\lceil \alpha w \rceil +1) + h(\lceil\alpha w \rceil -1)$, since each container $i$ with $i\leq r$ has height at most $h_r$ and each container $i$ with $i> r$ has height at most $h$. So in total we have
\[h_l(w-\lceil \alpha w \rceil +1) \leq A(C_B) = A(C_B')\leq h_r(w-\lceil \alpha w \rceil +1) + h(\lceil\alpha w \rceil -1). \]
Since $h_l > h_r$ and the difference between two items out of $T \cup P$ is at least $\beta$ we have $h_l \geq h_r +\beta$. This leads to 
\[(h_r +\beta)(w-\lceil \alpha w \rceil +1) \leq h_l(w-\lceil \alpha w \rceil +1) \leq h_r(w-\lceil \alpha w \rceil +1) + h(\lceil\alpha w \rceil -1).\]
It follows that 
$w\beta \leq (\beta +h)(\lceil\alpha w \rceil -1)$.
Since $\lceil\alpha w \rceil -1 < \alpha w$ this leads to 
$\beta < (\beta +h)\alpha$,
which is a contradiction for each $\alpha \leq \beta/(h +\beta )$.

\end{proof}

%
%

All the containers that can not be placed into the rearranged packings will be placed in an extra box $V_0$ of height $\tallItemHeight$ and width $(1-\alpha)W$. This alone with the knowledge from \cite{nadiradzeWiese} is enough to generate an algorithm with approximation ratio $4/3+\epsilon$. But we also want to speed up the algorithm by generating less boxes. Hence, we want the parameters $\alpha$ and $\delta$ as large as possible. 

\section{Improving the running time}
\label{runningTime}

The key to improve the running time, is to reduce the number of subboxes of a box in $\hatboxTV$. We do this by a other reordering of the tall items as in \cite{nadiradzeWiese}. In this section we will first present the new reordering strategy and then present some useful lemmas to improve the parameter $\delta$ additionally.

Let $S_P$ be the number of item sizes in $P$ and $S_T$ the number of item sizes in $T$ respectively. Additionally, let $S_{T\cup P}\leq S_T+S_P$ be the number of item sizes in $T\cup P$.
Furthermore, let us assume that there is at most one tall item on each side of the box, which overlaps the box border. In the following lemma we will present an algorithm wich reorders the tall and pseudo items, such that we generate few sub boxes.

\begin{lma}
\label{lma:partitionOflargeBoxes}
Let $B \in \hatboxTV$.
We can find a rearrangement of tall and pseudo items in $B$, such that we need at most $\mathcal{O}((S_P +S_T)S_{T\cup P})$ subboxes containing either tall or vertical items, such that each subbox $B_T$ for tall items contains just items with height $h(B_T)$, and all vertical items in $B$ can be packed fractionally into the subboxes for vertical items.
\end{lma}

\begin{proof}
We consider two cases. In the first no tall item overlaps the left or right border of $B$. For this case it is shown in \cite{nadiradzeWiese} that we can simply sort the items from $T \cup P$ touching the top of $B$ in descending order of heights and the items touching the bottom in ascending order. We sort tall and pseudo items of the same height such that pseudo items are positioned left to the tall items. By this reordering no two items overlap and we have at most $2S_P$ boxes for tall items and at most $2S_T$ boxes for pseudo items, summing up to $2(S_P+S_T)$ sub boxes total. 

In the second case, on each side can be one tall item, which overlaps the box border. In this case, we reorder the items differently from \cite{nadiradzeWiese}. We reduce the number of boxes from an exponential to a quadratic function in the number of different heights in $T \cup P$. 

Let $h_b$ be the height of a tallest item touching the bottom of the box and $b_l$ be the leftmost and $b_r$ the rightmost item of height $h_b$. Similarly choose $h_t, t_l,t_r$ with respect to the top.
Further, let $i_l$ be the item in $\{t_l, b_l\}$ which is further left and $i_r$ the item which is further right in $\{t_r, b_r\}$. 
If $i_l$ and $i_r$ are touching the same border we change $i_r$ to the other item in $\{t_r,b_r\}$ such that $i_l$ and $i_r$ touch different borders. Let w.l.o.g $i_l = b_l$ and $i_r = t_r$. 

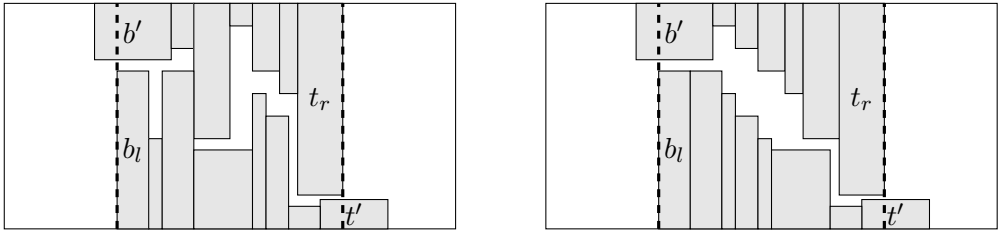
\begin{figure} [ht]
\begin{tikzpicture}
\pgfmathsetmacro{\h}{3}
\pgfmathsetmacro{\w}{6}
\draw (0,0) rectangle (\w,\h);
\draw[fill= black!10!white](0.25*\w,0) rectangle (0.32*\w,0.7*\h) node[midway] {$b_l$};
\draw[fill= black!10!white](0.65*\w,0.15*\h) rectangle (0.75*\w,\h) node[midway] {$t_r$};
\draw[fill= black!10!white] (0.7*\w,0) rectangle (0.85*\w,0.13*\h) node[midway] {$t'$};
\draw[fill= black!10!white] (0.2*\w,0.75*\h) rectangle (0.37*\w,\h) node[midway] {$b'$};
\draw[very thick, dashed] (0.25*\w,0) -- (0.25*\w,\h);
\draw[very thick, dashed] (0.75*\w,0) -- (0.75*\w,\h);
\draw[fill= black!10!white] (0.37*\w,0.8*\h) rectangle (0.42*\w,\h);
\draw[fill= black!10!white] (0.42*\w,0.4*\h) rectangle (0.5*\w,\h);
\draw[fill= black!10!white] (0.5*\w,0.9*\h) rectangle (0.55*\w,\h);
\draw[fill= black!10!white] (0.55*\w,0.7*\h) rectangle (0.61*\w,\h);
\draw[fill= black!10!white] (0.61*\w,0.6*\h) rectangle (0.65*\w,\h);
\draw[fill= black!10!white](0.32*\w,0) rectangle (0.35*\w,0.4*\h);
\draw[fill= black!10!white](0.35*\w,0) rectangle (0.42*\w,0.7*\h);
\draw[fill= black!10!white](0.42*\w,0) rectangle (0.55*\w,0.35*\h);
\draw[fill= black!10!white](0.55*\w,0) rectangle (0.58*\w,0.6*\h);
\draw[fill= black!10!white](0.58*\w,0) rectangle (0.63*\w,0.5*\h);
\draw[fill= black!10!white](0.63*\w,0) rectangle (0.7*\w,0.1*\h);

\begin{scope}[xshift = 1.2*\w cm]
\draw (0,0) rectangle (\w,\h);
\draw[fill= black!10!white](0.25*\w,0) rectangle (0.32*\w,0.7*\h) node[midway] {$b_l$};
\draw[fill= black!10!white](0.65*\w,0.15*\h) rectangle (0.75*\w,\h) node[midway] {$t_r$};
\draw[fill= black!10!white] (0.7*\w,0) rectangle (0.85*\w,0.13*\h) node[midway] {$t'$};
\draw[fill= black!10!white] (0.2*\w,0.75*\h) rectangle (0.37*\w,\h) node[midway] {$b'$};
\draw[very thick, dashed] (0.25*\w,0) -- (0.25*\w,\h);
\draw[very thick, dashed] (0.75*\w,0) -- (0.75*\w,\h);
\draw[fill= black!10!white] (0.37*\w,0.9*\h) rectangle (0.42*\w,\h);
\draw[fill= black!10!white] (0.42*\w,0.8*\h) rectangle (0.47*\w,\h);
\draw[fill= black!10!white] (0.47*\w,0.7*\h) rectangle (0.53*\w,\h);
\draw[fill= black!10!white] (0.53*\w,0.6*\h) rectangle (0.57*\w,\h);
\draw[fill= black!10!white] (0.57*\w,0.4*\h) rectangle (0.65*\w,\h);
\draw[fill= black!10!white](0.32*\w,0) rectangle (0.39*\w,0.7*\h);
\draw[fill= black!10!white](0.39*\w,0) rectangle (0.42*\w,0.6*\h);
\draw[fill= black!10!white](0.42*\w,0) rectangle (0.47*\w,0.5*\h);
\draw[fill= black!10!white](0.47*\w,0) rectangle (0.5*\w,0.4*\h);
\draw[fill= black!10!white](0.5*\w,0) rectangle (0.63*\w,0.35*\h);
\draw[fill= black!10!white](0.63*\w,0) rectangle (0.7*\w,0.1*\h);
\end{scope}
\end{tikzpicture}
\caption{A packing before and after the reordering of the items.}
\label{fig:smallPacking}
\end{figure}

We draw a vertical line at the left border of $b_l$. The item we cut with this line we define as a new unmovable item $b'$. We do the same on the right side of $t_r$ and name the cut item $t'$ (see figure \ref{fig:smallPacking}).
Now we sort the movable items between the drawn vertical lines. 
The movable items touching the top are sorted in ascending order with respect to their height, while the movable items touching the bottom are sorted in descending order. 

We will show now, that in this reordering no two items overlap.
There is no tall item touching the bottom that overlaps $b'$ since each item touching the bottom has height at most $h_b$. Since $b'$ was placed above $b_l$ this means $b'$ fits above each item in the box $B$. Similarly one can see that no item overlaps $t'$.

Assume now there is an item $i_b$ touching the bottom that overlaps an item $i_t$ touching the top. Let $p = (x_p,y_p)$ be a point, which is overlapped by the item $i_b$ and $i_t$. Let $(x_l,y_b)$ denote the left bottom corner of $b_l$ and $(x_r,y_t)$ the right top corner of $t_l$. By our reordering there must be a set of items $I_b$ touching the bottom with total width greater than $x_p - x_l$, which is placed between $x_l$ and $x_r$ and has height at least $y_p - y_b$.  Furthermore there must be a set of items $I_t$ with total width greater than $x_r - x_p$ touching the bottom and having height at least $y_t - y_p$. 
Since the area the items can be placed in has a width of $x_r -x_l$ and the sets $I_b$ and $I_t$ have a total width of $w(I_t \cup I_b) > x_p - x_l+ x_r - x_p = x_r - x_l$ by the pidgin hole principle there must be an item in $I_b$ that overlaps an item in $I_t$ in the original packing.

We now look at the items touching the top and having the same height as $b'$. We remove this set of items, shift the items smaller than $h(b')$ to the right and place the items with height $h(b')$ next to $b'$. After this shifting no two tall items overlap. This is because we have shifted the smaller items above items, which are shorter than the items they were placed above before. By this shifting, we avoid that we need an extra box for the item $b'$. We do the same on the bottom with the items with height $h(t')$.

So far we have achieved the following: We have at most $2S_T$ boxes for tall items between $i_l$ and $i_r$ and at most $2S_P$ for pseudo items.
The total number of different heights touching the bottom and touching the top, on the left of $i_l$ is at least one smaller than in the whole box. Same holds for the number on the right side of $i_r$.

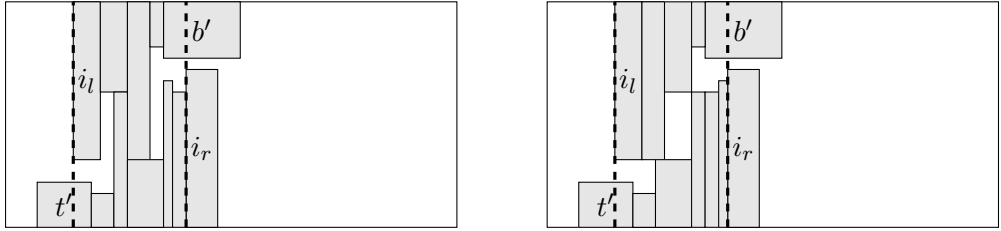
\begin{figure} [ht]
\centering
\begin{tikzpicture}
\pgfmathsetmacro{\h}{3}
\pgfmathsetmacro{\w}{6}
\draw (0,0) rectangle (\w,\h);
\draw[fill= black!10!white](0.4*\w,0) rectangle (0.47*\w,0.7*\h) node[midway] {$i_r$};
\draw[fill= black!10!white] (0.35*\w,0.75*\h) rectangle (0.52*\w,\h) node[midway] {$b'$};
\draw[fill= black!10!white] (0.15*\w,0.3*\h) rectangle (0.21*\w,\h)node[midway] {$i_l$};
\draw[fill= black!10!white] (0.07*\w,0.2*\h) rectangle (0.19*\w,0)node[midway] {$t'$};

\draw[fill= black!10!white] (0.21*\w,0.6*\h) rectangle (0.27*\w,\h);
\draw[fill= black!10!white] (0.27*\w,0.3*\h) rectangle (0.32*\w,\h);
\draw[fill= black!10!white] (0.32*\w,0.8*\h) rectangle (0.35*\w,\h);
\draw[fill= black!10!white] (0.19*\w,0.15*\h) rectangle (0.24*\w,0);
\draw[fill= black!10!white] (0.24*\w,0.6*\h) rectangle (0.27*\w,0);
\draw[fill= black!10!white](0.27*\w,0) rectangle (0.35*\w,0.3*\h);
\draw[fill= black!10!white](0.35*\w,0) rectangle (0.37*\w,0.65*\h);
\draw[fill= black!10!white](0.37*\w,0) rectangle (0.4*\w,0.6*\h);
\draw[very thick, dashed] (0.4*\w,0) -- (0.4*\w,\h);
\draw[very thick, dashed] (0.15*\w,0) -- (0.15*\w,\h);
\begin{scope}[xshift = 1.2*\w cm]
\draw (0,0) rectangle (\w,\h);
\draw[fill= black!10!white](0.4*\w,0) rectangle (0.47*\w,0.7*\h) node[midway] {$i_r$};
\draw[fill= black!10!white] (0.35*\w,0.75*\h) rectangle (0.52*\w,\h) node[midway] {$b'$};
\draw[fill= black!10!white] (0.15*\w,0.3*\h) rectangle (0.21*\w,\h)node[midway] {$i_l$};
\draw[fill= black!10!white] (0.07*\w,0.2*\h) rectangle (0.19*\w,0)node[midway] {$t'$};

\draw[fill= black!10!white] (0.21*\w,0.3*\h) rectangle (0.26*\w,\h);
\draw[fill= black!10!white] (0.26*\w,0.6*\h) rectangle (0.32*\w,\h);
\draw[fill= black!10!white] (0.32*\w,0.8*\h) rectangle (0.35*\w,\h);
\draw[fill= black!10!white] (0.19*\w,0.15*\h) rectangle (0.24*\w,0);
\draw[fill= black!10!white](0.24*\w,0) rectangle (0.32*\w,0.3*\h);
\draw[fill= black!10!white] (0.32*\w,0.6*\h) rectangle (0.35*\w,0);
\draw[fill= black!10!white](0.35*\w,0) rectangle (0.38*\w,0.6*\h);
\draw[fill= black!10!white](0.38*\w,0) rectangle (0.4*\w,0.65*\h);
\draw[very thick, dashed] (0.4*\w,0) -- (0.4*\w,\h);
\draw[very thick, dashed] (0.15*\w,0) -- (0.15*\w,\h);
\end{scope}
\end{tikzpicture}
\caption{A recursive rearrangement of the tall and pseudo items.}
\end{figure}

We now describe how to continue to reorder the packing: We repeat the following step until a break condition occur. In each step, we will reduce the total number of different heights of the items touching the top and bottom by at least one.
We look on the left side of $i_l$. W.l.o.g. let $i_l$ touching the bottom of the box. Let $b'$ be the item, which was intersected by the vertical line at the left border of $i_l$. 
Let $h_t$ be the height of the largest item touching the top left of $i_l$. We rename the item $i_l$ as $i_r$ and  redefine $i_l$ as the left most item touching the top, which has height $h_t$. Again we draw a vertical line on the left side of $i_l$. Let $t'$ be the item intersected by this line.  Again we consider $t'$ and $b'$ as unmovable items. We sort the movable items touching the bottom between $t'$ and $\bar{i}_L$ in ascending order and the movable items touching the top in descending order. With the same arguments as above, one can see that by this reordering no item from the bottom overlaps an item from the top.

By choosing $i_l$ as the leftmost tallest item touching the top we have reduced the total number of different heights touching the top and bottom in the remaining area, which has to be reordered, by at least one. We repeat the described step until one of the following conditions occur:
\begin{enumerate}
\item The tallest item touching the top and the tallest item touching the bottom have a summed height of at most $h$.
\item The item $i_r$ is the unmovable item, which overlaps the left border.
\end{enumerate}

If condition 1. occur in any reordering of the items it can not happen that a tall or pseudo item touching the bottom overlaps any tall or pseudo item touching the top, since their height is not large enough. So at this point we simply sort the items touching the top in ascending order and the items touching the bottom just as well.

If condition 2. occur we repeat the normal reordering step once again. When we draw the vertical line, it will be placed exactly on the box border, and we are finished.

We repeat this steps analog on the right side of the initial $i_r$. Let us consider how many  different subbox for tall and pseudo items we create by this step: As seen before in each of the partitioning steps we create at most $2 S_T$ subboxes for tall items and at most $2S_P$ subboxes for pseudo items. 

In each of the partition steps, we reduce the total number of different heights touching the bottom and the top by one. 
If the tallest item touching the top and the tallest item touching the bottom are both smaller than $h(B)/2$, then condition 1 is fulfilled. Since in each partitioning step we reduced the problem by one of the tallest item sizes we need at most $S_{T\cup P}$ steps until the tallest item touching the bottom and the tallest item touching the top both have a height of at most $h(B)/2$. 

Since we reorder iteratively to the left and to the right, we create at most $4S_{T\cup P}S_T$ boxes for tall items and at most $4S_{T\cup P}S_P$ boxes for pseudo items in total.
 

\end{proof}

To this point it holds that  $S_T$, $S_P$ and $S_{T\cup P} \in \mathcal{O}(1/\eps^2)$.
In the following lemma we reduce the values of $S_T$, $S_P$ and $S_{T\cup P}$ and assure that we can use the algorithm form lemma \ref{lma:partitionOflargeBoxes}, by providing the needed properties. As a consequence of Lemma \ref{lma:heightOfBeta} we get $S_T, S_P,S_{T\cup P} \in \mathcal{O}(1/\eps)$. Leading to a partition into at most $\mathcal{O}(1/\eps^2)$ subboxes of each box in $\hatboxTV$.

\begin{lma}
\label{lma:heightOfBeta}
If we increase the height of the packing area $W \times (1+2\eps)\OPT$ by $3\eps\OPT$, we can assume that each item in $T \cup P$ has a height, which is a multiple of $\eps \OPT$. Furthermore, at each side of a box, there is at most one tall item overlapping its border, that touches either its bottom or top and has a height, that is a multiple of $\eps\OPT$. By this step we intruduce $\mathcal{O}(1)$ subboxes for each box in $\hatboxTV$.
\end{lma}

We will prove this lemma in three parts: First, we will show, that we can ensure that each tall item has a height, which is a multiple of $\eps \OPT$ when we add $2\eps \OPT$ to the packing height. After that, we show that we can guarantee that each box in $\hatboxTV$ has a height, which is a multiple of $\eps \OPT$ when we enlarge the packing height by $\eps \OPT$. This ensures that the generated pseudo items have a height which is a multiple of $\eps \OPT$ as well. In the last step, we will look at the unmovable items.
\begin{lma}
\label{lma:tallItemHeight}
At a loss of an approximation ratio of at most $2 \eps \OPT$ we can ensure, that each tall item has a size, which is a multiple of $\eps \OPT$. Each box for vertical items of height at least $(2/3 + 2\eps) \OPT$ will be enlarged by at least $2 \eps \OPT$. 
\end{lma}
\begin{proof}
Since a tall item has a height $h_i > (1/3 +\eps)\OPT$ each tall item is either intersects the horizontal line at $(1/3 +\frac{2}{3}\eps)\OPT$ or the horizontal line at $(2/3 +\frac{4}{3}\eps)\OPT$. 
We shift all items that start after $(2/3 +\frac{4}{3}\eps)\OPT$ exactly $\eps \OPT$ upwards. By this shifting, all horizontal and large boxes above $(2/3 +\frac{3}{2}\eps)\OPT$ and the items in it stay unchanged, except that they have moved upwards. The vertical boxes starting before $(2/3 +\frac{3}{2}\eps)\OPT$ and ending after or at $(2/3 +\frac{3}{2}\eps)\OPT$ are enlarged by $\eps \OPT$. Notice that there is no tall item starting after $(2/3 +\frac{4}{3}\eps)\OPT$ since each tall item has a height which is at least $\tallItemHeight$ and the packing has a height of at most $(1+2\eps)\OPT$. Now above each tall item ending between $(2/3 +\frac{4}{3}\eps)\OPT$ and $(1+2\eps)\OPT$ is a gap of height $\eps\OPT$. We use this gap to extend each tall item until it has a height, which is a multiple of $\eps\OPT$. More precisely we round the items height to $\lceil h_i/\eps\OPT \rceil \eps\OPT$.

Since each tall item has height of at least $\tallItemHeight$ we know that each tall item, which has not jet a height, which is a multiple of $\eps \OPT$, ends between $(1/3 +\eps)\OPT$ and  $(2/3 +\frac{4}{3}\eps)\OPT$. Since the tall item ends before $(2/3 +\frac{4}{3}\eps)\OPT$ the latest possible start time is $(1/3 +\frac{2}{3}\eps)\OPT$.
So we shift all items starting after $(1/3 +\eps)\OPT$ exactly $\eps \OPT$ upwards. As in the step before we do not create any new box, but we enlarge all boxes for vertical and tall items starting before $(1/3 +\eps)\OPT$ end ending after or at $(1/3 +\eps)\OPT$. This step creates a gap of height $\eps \OPT$ above each tall item, starting before $(1/3 +\eps)\OPT$. So we can stretch all tall items such that they have a height which is a multiple of $\eps\OPT$. 

Notice that boxes of height at least $(2/3 + 2\eps)\OPT$ start before $1/3 \cdot \OPT$ end end after $(2/3 + 2\eps)\OPT$ so they are enlarged by both shifting operations.  
\end{proof}

Now $S_T \leq 1/\eps$ and each box containing two tall items above each other have a height of at least $(2/3 + 4\eps)\OPT$. We still need that all pseudo items, which will be generated in this boxes have a height, which is a multiple of $\eps\OPT$. If the box has a height which is a multiple of $\eps\OPT$ and all tall items have a height which is a multiple of $\eps\OPT$ as well, the property that each tall item has a height, which is a multiple of $\eps \OPT$ follows by the construction of the tall items.

\begin{lma}
\label{lma:pseudoItemheight}
At loss of at most $\eps \OPT$ in the approximation ratio, we can ensure that each box in $\hatboxTV$ has a height, which is a multiple of $\eps \OPT$.
\end{lma}
\begin{proof}
We shift each box, with left bottom coordinate $(x,y)$ and $y \geq(2/3 + 4\eps)\OPT$ exactly $\eps\OPT$ upwards.
We can do this shifting operation since we are allowed to slice the vertical items. By this shifting, no tall item is sliced, since all of them are starting before $(2/3 + 2\eps)\OPT$. So all boxes containing these items do so also. So none of this boxes is shifted. 

Since each box with height at least $(2/3 + 4\eps)\OPT$ has a start point below $(2/3 + 4\eps)\OPT$ and has an upper y-coordinate, which is at least $(2/3 + 4\eps)\OPT$
we have a free space of at least $\eps \OPT$ above this box. So we can enlarge this box by at most $\eps \OPT$, such that its height is a multiple of $\eps \OPT$.
\end{proof}

The packing we consider has now a height of at most $(1+5\eps)\OPT$. The last thing that could destroy the property that all tall and pseudo items we are going to rearrange, have a height, which is a multiple of $\eps \OPT$ are the unmovable items. Luckily we can extend the unmovable items in such a way that they do not destroy this property:

\begin{lma}
\label{lma:overlappingItems}
Let $B$ be a box of height $h \geq (2/3 +4\eps)\OPT$. We can assume that at each side of the box there is at most one tall item overlapping this box. This tall item touches either the bottom or the top of the box and has a height, which is a multiple of $\eps\OPT$. We introduce at most $8$ container for vertical items, to guarantee this property.
\end{lma}

\begin{proof}
Let us consider one side of the box where two tall items overlap the box border. First, we can assume that these items are touching the top and the bottom of the box. 
If they do not touch the bottom or the top, we can enlarge the items such that they do, by introducing one unmovable container containing the vertical items, which are positioned between the box border and the overlapping item. The overlapping item is then redefined as the item consisting of the overlapping item and the container for vertical items. Furthermore, we can assume that there is just one tall item overlapping the box border: Let $i$ and $j$ be the two items overlapping the border (both are touching either the top or the bottom of the box). 
Let us look at the left side of the box. One of the items right border is positioned at a higher x position than the other item. 
Let w.l.o.g. $j$ be this item. We redefine the left border of the box $B$ such that it is positioned at the x coordinate where the right border of the item $i$ is positioned. We introduce a container for the vertical items positioned between the two overlapping items. By this operation, we have created at most $3$ container for vertical items on each side of the box.

The height of the vertical item overlapping the border could be a value which is not a multiple of $\eps\OPT$ since we had glued the container for vertical items to it. Let us w.l.o.g. assume that the overlapping item touches the bottom of the box. We have two cases: there either is a tall item above the overlapping item or there is a pseudo item above the overlapping item. If there is a pseudo item, we glue the overlapping item to the pseudo item. By this step, we generate an item with height $h$. So this item has clearly a height which is a multiple of $\eps \OPT$ since the box has a height, which is a multiple of $\eps \OPT$. Furthermore, we can now assume that no tall item overlaps this border of the box since we can shift the box border such that the item is no longer contained in this box. The box border would now intersect no other tall item. 

If there is a tall item $t$ touching the top, it has a height, which is a multiple of $\eps \OPT$. Between the overlapping item and the tall item touching the top, there can be just vertical items. We generate a container for the vertical items in the area between the overlapping item and $t$. If we combine the overlapping item with the container, we get a new unmovable item with height $h - h_t$, which is a multiple of $\eps \OPT$.
\end{proof}

All these steps together ensure the properties from Lemma \ref{lma:heightOfBeta}.
Note that for each box in $\hat{\mathcal{B}}_{T \cup V}$ we introduce $8$ subboxes containing vertical items, to guarantee the properties above.
Since now each item height in $T \cup P$ is a multiple of $\eps\OPT$ we have $\beta \geq \eps \OPT$. Let us take a look at items that are very tall with respect to the size of a box $B \in \hatboxTV$. Consider an item $i$ with height larger than $h(B) - (1/3 +\eps)\OPT$. Since each tall item has height larger than $\tallItemHeight$, there can be no tall item placed above or below this item.   
By construction, there is one pseudo item directly above or below $i$.
We combine $i$ and the pseudo item to one new pseudo item which has height $h(B)$ and width $w_i$. Now it holds that the distance between items touching the bottom and items touching the top is at most $(1/3 + 3\eps)\OPT$. So we can choose $\alpha = 2\eps < \eps/(1/3+4\eps)$ for $\eps < 1/24$. There are now at most $(1/3 + 3\eps)/\eps +2 = (1/3 + 5\eps)/\eps$ possible item sizes in $P \cup T$ with respect to a box $B \in \hat{\mathcal{B}}_{T \cup V}$. Therefore we get $S_T, S_P , S_{T\cup P} \leq (1/3 + 5\eps)/\eps$. 

In the next step, we look at the boxes in $\checkboxTV$ and their partitioning.

\begin{lma}
\label{lma:ShortBoxes}
We can find a rearrangement of the items in each box $B \in \checkboxTV$ such that we can partition the area in $B$ into at most $\mathcal{O}(1/\eps)$ subboxes for tall items and at most $\mathcal{O}(1/\eps)$ subboxes for vertical items, such that all vertical items can be packed fractionally in these subboxes, and each subbox for tall items contains just items with the same height.
\end{lma}

We show this in two steps. To rearrange the items in a box $B \in \checkboxTV$, we first shift the tall items down, such that they touch the bottom of the box. After that we sort the items touching the bottom, such that items with the same height are positioned next to each other. By this rearrangement, no vertical item has to be placed outside of the box, but we maybe have to slice some of the vertical items. That we can rearrange the items in this way is stated in the following two Lemmas.

\begin{lma}
If we are allowed to slice the items in $V$ vertically, we can ensure that each tall item in a box $B \in \checkboxTV$ is touching the bottom of the box.
\end{lma}
\begin{proof}
It is not possible that a vertical line through the box $B$ intersects two tall items, since each tall item has a height of at least $(1/3 +2\eps)\OPT$. We now look at one tall item $t$, which is not crossed by the border of the box. We draw vertical lines at the left and the right border of the item and slice all items which are crossed by this line in this box. As noticed before this are only vertical items. Now we have below $t$ a small box, which borders are the vertical lines, the bottom of $B$, and the bottom of the item $t$. There is no vertical item in this box, which is crossed horizontally by the box borders. We remove this box and all the items it contains, shift $t$ down such that it touches the bottom of $B$ and place the small box with all the vertical items on top of $t$. We repeat this step with all tall items, which are not crossed by the box border. So now all movable items touch the bottom of the box.
\end{proof}

\begin{lma}
\label{lma:rearrangementInSmallBox}
Let $B \in \checkboxTV$. Then there is a rearrangement of the items in $B$ such that there are at most $\mathcal{O}(1/\eps)$  different container for tall items, and at most $\mathcal{O}(1/\eps)$ different container for vertical items.
\end{lma}
\begin{proof}
We define vertical slices, by drawing vertical lines at each side of the tall items. 
We define the slides containing no unmovable item as movable slides. 
We sort the slides in decreasing order of the height of the tall rectangle they contain. 
By this reordering, we get at most $1/\eps$ container for movable tall items, since they have a height, which is a multiple of $\eps \OPT$. 
There are at most two tall items overlapping the left or the right border of $B$. So for these items, we need at most two extra boxes.

Above each container for tall items, we introduce one container for vertical items. For the tall items overlapping the left and right border, we have to introduce at most $4$ container. Additionally, it can happen that we have to introduce one container having the height $h(B)$, where no tall item is positioned. So in total, we have at most $\mathcal{O}(1/\eps)$ container for vertical items.
\end{proof}

Since the boxes in $\checkboxTV$ can be partitioned into less boxes than the boxes $\hatboxTV$, the following Lemma follows from Lemma \ref{lma:verticalItemShift}, Lemma \ref{lma:heightOfBeta} and Lemma \ref{lma:partitionOflargeBoxes}. 

\begin{lma}
\label{lma:NumberBoxes}
We can partition boxes $\mathcal{B}_{T\cup V}$ such that we introduce at most $\mathcal{O}(1/\eps^3\delta^2)$ boxes for tall items $\boxT$, each containing just items with the same height, and at most $\mathcal{O}(1/\eps^3\delta)$ boxes $\boxV$ for vertical items, such that all vertical items can be packed fractionally into the the boxes $\boxV$ and an additional box $V_0$ with height $\tallItemHeight$ and width $(1-2\eps)W$.
\end{lma}

\begin{proof}
We have to show two things: the number of boxes is as small as required and that the vertical items can be placed into the boxes $\boxV \cup {V_0}$.

Since we have at most $3(1+4\eps)/(\eps\delta^2)$ different boxes in $\mathcal{B}_{T\cup V}$ and each box is partitioned into at most $4((1/3 + 5\eps)/\eps)^2 +8$ container for tall and pseudo items respectively, we generate at most $(3(1+4\eps)(4(1/3 + 5\eps)^2 +8\eps^2))/(\eps^3\delta^2) = 4(1+34\eps+363\eps^2 + 972\eps^3)/3\eps^3\delta^2=\mathcal{O}(1/(\eps^3\delta^2))$ container for pseudo and tall items respectively in total. More precisely for $\eps < 1/26$ we have $4(1+34\eps+363\eps^2 + 972\eps^3)/3\eps^3\delta^2 \leq 4/\eps^3\delta^2$.

We now define the set of subboxes for vertical items $\boxV$. For each subbox for pseudo items, we introduce one subbox for vertical items, which has height and width of the box for the pseudo items. Furthermore, we introduce subboxes for vertical items, which are positioned in the area between two subboxes for items out of $T \cup P$, which are positioned in a vertical line in the same box. We do this by drawing a vertical line at each subbox left border. These lines partition the region between the pseudo and tall items. Since we draw one line per subbox, each subbox generate at most one part of the partition of the area. So we have partitioned the area into at most $4/(\eps^3\delta^2)$  parts. Each of this parts is a rectangular area and defines a new subbox for vertical items.
So in total, we have now at most $8/(\eps^3\delta^2)$ subbox for vertical items. We denote this set of subbox for vertical items as $|\mathcal{B}_V|$.

We know that the vertical items that are overlapped by the pseudo items can be fractionally placed into the boxes $\boxV$. By Lemma \ref{lma:verticalItemShift} we know that we can place at least $2\eps W$ of the container for the vertical items in $B_V$ can be placed into the boxes. Therefore all items that can not be placed into the boxes $\boxTV$ can be placed into the box $V_0$ since these subboxes have height at most $(1/3+\eps)\OPT$ and a total width, which is smaller than $(1-2\eps)W$.
\end{proof}

We can sum up the structural result in the following Lemma:

\begin{lma}[Structural Lemma]
\label{lma:structuralLemma}
By slicing vertical items each optimal packing can be rearranged, such that the  packing area $W \times (4/3 +6\eps)\OPT$ can be partitioned into boxes with the following properties:
\begin{itemize}
\item All small and medium sized items are removed.
\item Each item in $i \in L$ is contained in a box from $\boxL$, which has height $h_i$ and width $w_i$.
\item There are at most $(1+2\eps)/(\eps\delta^2) - |L|/\delta$ boxes in $\boxH$, which contain all horizontal items.
\item There are at most $4/\eps^3\delta^2$ boxes $\boxT$ containing all tall items, such that each box $B \in \boxT$ just contains tall items with height $h(B)$.
\item There are at most $8/\eps^3\delta^2$ boxes $\boxV$ and one box $V_0$ of height $(1/3+\eps)\OPT$ and width $(1-2\eps)W$, containing all (sliced) vertical items.
\item the free area in the boxes $\boxH \cup \boxV \cup \{V_0\}$ is at least $AREA(S) + (1-2\eps)(1/3+\eps)\OPT W$
\end{itemize}

\end{lma}

\section{Algorithm}
\label{Algorithm}
The following Lemma states that we can find a packing of the items into a partition from Lemma \ref{lma:structuralLemma} in polynomial time. We define boxes $B_{H}$ with size $W \times 2\eps\OPT$, $B_{MH}$ with size $W \times \eps\OPT$, $B_{MV}$ with size $3\eps W/2 \times (1/3 +\eps)\OPT$ and $B_{V}$ with size $(1-3\eps/2)W \times (1/3 +\eps)\OPT$.

\begin{lma}
\label{lma:algorithmToPackTheItems}
Let a partition into boxes form Lemma \ref{lma:structuralLemma} be given. 
There is an algorithm with running time $\mathcal{O}(n \log n + W^{(1/\eps\delta)^{\mathcal{O}{(1)}}})$ that packs all the items in $I$ into the boxes $\boxL \cup \boxH \cup \boxT \cup \boxV \cup \{B_{H}, B_{MH}, B_{MV}, B_{V}\}$ or decides that such packing does not exist. 
\end{lma}

We will prove this Lemma in $4$ steps. First, we will show that the medium sized items can be placed into their two boxes $B_{MV}$ and $B_{MH}$. Then we will look at the horizontal items. After that, we will focus on the vertical items. We will show that by placing horizontal and vertical items, we leave enough free area to place small items. To place the tall items we use a result by Nadiradze and Wiese \cite{nadiradzeWiese}. By generating the packing, we can assume that $\eps \leq 24$.

\begin{lma}
\label{lma:mediumsizedItems}
The set of medium sized horizontal items $M_H$ can be placed in a rectangular area with width $W$ and height $\eps\OPT$. 
The set of medium sized horizontal items can be placed in a rectangular area of width $3/2 \cdot \eps W$ and height $(1/3 +\eps)\OPT$. The algorithm to place this items has a running time of $\mathcal{O}(n\log(n))$.
\end{lma}
\begin{proof}
We know that $A(M_H) \leq f(\eps) W \OPT = \eps /6 \cdot W \OPT$ and for each $i \in M_H$ we have that $w_i \leq W$.  We know by \cite{Coffman} that if we pack these items with the NFDH algorithm into a strip with width $W$ the packing height is at most $2A(M_H)/ W + h_{max} \leq (\eps/ 3 +\delta)\OPT \leq \eps \OPT$.

The items in $M_V$ will be placed with the NFDH algorithm as well, but this time we rotate the items and the packing area by 90 degree first. We now pack the rotated item into a strip with width $(1/3 + \eps)\OPT$. Since the items have a total area of at most $A(M_V) \leq \eps /6 \cdot W \OPT$ and all the items have a width of at most $(1/3 + \eps)\OPT$, \cite{Coffman} implies that we can place the items in the strip, constructing a packing with height at most $2A(M_V)/(1/3 + \eps)\OPT + h_{max} = 2\eps W /6(1/3 + \eps)+\delta W= \eps W/(1+3\eps) +\delta W\leq 3/2 \cdot \eps W$. 
\end{proof}

This proves that we can place the medium sized items into the boxes $B_{MH}$ and $B_{MV}$. We will now see how the horizontal items are placed. The main idea is to place them in their boxes with a linear programming approach as seen many times before. We will show that if $\boxH$ is the set of boxes obtained from the optimal packing, we can pack the items in $H$ in $\mathcal{B}_H \cup \{B_{H0}\}$. We use the fact that the items fractionally fit into the boxes. Here fractionally means, that we just need to slice the items horizontally to find a valid packing in the boxes. When we pack the items with the algorithm described in the proof, we will generate a constant number of boxes for small items. Their area has at least the size of the area the small items in the boxes $\boxH$ had used before.

\begin{lma}
\label{lma:packingHorizontalItems}
There is a polynomial time algorithm that assigns all items in $H$ into the boxes $\mathcal{B}_H \cup \{B_{H}\}$ if there is a fractional packing of these items into the boxes $\mathcal{B}_H$. The algorithm needs at most $\mathcal{O}((1/\eps\delta)^{(1/\delta)}) +\mathcal{O}(n\log(n))$ operations. This algorithm generates at most $1/\eps^2\delta^2$ Boxes $\mathcal{B}_{S,H}$ boxes for small items with total area at least $AREA(\mathcal{B}_H) - AREA(H)$.
\end{lma}
\begin{proof}
We use the fact that it is possible to place all horizontal items into the boxes. We know this since all parts of the horizontal items are completely overlapped by the boxes. In the first step, we do a linear grouping step to round the width of the horizontal items to at most $1/\eps\delta$ different widths. We do this by stacking all horizontal items on top of each other by ascending width. This stack has a height of at most $\frac{1}{\delta}\OPT$ since each item has a width of at least $\delta W$, and the total area of horizontal items is at most $W\OPT$. We now draw horizontal lines at each multiple of $\eps \OPT$ and split each horizontal item which is cut by this line. The items between two of this lines define a group of items. We have generated at most $1/\eps\delta$ of these groups since $\frac{1}{\delta}\OPT / (\eps \OPT) = 1/\eps\delta$. 
We now round up the size of each item in each group to the size of the largest item in its group. As seen in \cite{KarmarkarKarp} we can place all rounded horizontal items fractionally into the boxes except for the widest group. The widest group is placed in the extra box $B_H$. This last group has a total processing time of at most $(\eps + \mu)\OPT$.

To find an assignment of the horizontal items to the horizontal boxes we can solve the following LP. A configuration $C$ for a box $B$ is a multi set of rounded items, whose summed width does not exceed the width of the box $B$. We denote by $C(w)$ the number of items with width $w$ contained in $C$. Since all the horizontal items have width at least $\delta$ each configuration contains at most $1/\delta$ items. Since we have at most $1/\eps\delta$ different item widths we have at most $(1/\eps\delta +1)^{1/\delta}$ different configurations for each box. We denote by $\mathcal{C}_B$ the set of configurations for the box $B$. Let $h(w)$ be the total height of all items with width $w$. For $i \in \{1, \dots, 1/\eps\delta -1\}$ let $w_i$ be the width of the items in group $I$. The configuration LP has now the following form
\begin{align*}
\sum_{B \in \mathcal{B}_H}\sum_{C \in \mathcal{C}_B} C(w_i) x_{C,B} & = h(w_i) & \forall i \in \{1, \dots, 1/\eps\delta -1\}\\
\sum_{C \in \mathcal{C}_B}x_{C,B} & \leq h(B) & \forall B \in \mathcal{B}_H\\
x_{C,B} & \geq 0 & \forall B \in \mathcal{B}_H, C \in \mathcal{C}_B
\end{align*}

Since it is possible to place all horizontal items into the boxes this LP has a solution. This LP has $(1/\eps\delta +1)^{1/\delta} \cdot 1/\eps\delta^2$ variables and $1/\eps\delta + (1+2\eps)/\eps\delta^2$ conditions. So we can find a basic solution, which has $1/\eps\delta + (1+2\eps)/\eps\delta^2 \leq 2/\eps\delta^2$ non zero components in at most $\mathcal{O}((1/\eps\delta)^{1/\delta})$ operations. So in total we use at most $1/\eps\delta + 1/\eps\delta^2$ different configurations.
We fill this configurations greedily with the original items, such that the topmost item is allowed to overlap the configuration border. By an area argument one can see that it is possible to place all horizontal items in this way. 

Now in each configuration for each occurrence of an item size, we remove the last added item from the configuration. Now the fill height of this configuration does not exceed the reserved height of this configuration. All items we have removed form one configuration fit next to each other in the strip, since the configuration was feasible. 
So we can place them next to each other on top of the largest group of items into the box $B_{H0}$. 
Since the items in $H$ have height at most $\mu$ we add per configuration a layer of  height at most $\mu$ to the stack in the box. Since we have at most $2/\eps\delta^2$ configurations, the stack has a height of at most $2\mu/\eps\delta^2 +\leq \eps-\mu$, since $\mu \leq \eps^2\delta^2/2 +1\leq \eps^6\delta^2$. Therefore, the total packing height in $B_H$ is at most $2\eps$.

The configurations we place into the boxes are placed such that they touch the left border of the box. So it can happen that between the right border of the configuration and the right border of the box is some free area. 
This free area builds a container for small items. Since we have at most $1/\eps\delta + (1+2\eps)/\eps\delta^2$ configurations and at most $(1+2\eps)/\eps\delta^2$ boxes, there are at most $(2+4\eps +\delta)/\eps\delta^2 \leq 1/\eps^2\delta^2$ boxes $\mathcal{B}_{S,H}$ for small items. 
Since the solution to the LP satisfies the first condition, the area used by the configurations is at most $AREA(H)$. Therefore the boxes have an area of at least $AREA(\mathcal{B}_H) - AREA(H)$.
\end{proof}

To place the vertical items we use the same strategy as to place the horizontal items. We know from Lemma \ref{lma:NumberBoxes} that the vertical items can be fractionally placed in the boxes $\boxV$ plus an additional one with area $(1-2\eps)W \times \tallItemHeight$. Fractionally means here, that the items are allowed to be sliced vertically.

\begin{lma}
\label{lma:packingVerticalItems}
There is a polynomial time algorithm that places all vertical items into the boxes $\mathcal{B}_V \cup \{B_{V}\}$, if there is a fractional packing for these items into the boxes $\mathcal{B}_V \cup \{B_{V0}\}$. The algorithm needs at most $\mathcal{O}((1/\delta)^{1/\eps\delta})$ operations. This procedure creates at most $9/(\eps^3\delta^2)$ container $\boxSV$ for small items.
The total area of the container $\boxSV$ is at least $Area(\mathcal{B}_V\cup \{B_{V}\}) - Area(V)$.
\end{lma}

\begin{proof}
Consider all the boxes in $\mathcal{B}_V$. Each of this boxes $B$ has a height $h(B)$ and a width $w(B)$. We are interested in the total width of all boxes which have a specific height $h$. We denote that width with $w_{\mathcal{B}_V}(h)$. Let $H_B := \{i\eps\delta \OPT | i = 1, \dots , 1/\eps\delta\} \cup \{i \eps \OPT | i = 1, \dots, (1 +5\eps)/\eps\}$ the set of different container heights (container with height $> \OPT$ are generated by pseudo items and have a height, which is a multiple of $\eps\OPT$). Let $H_I:= \{i\eps^k \OPT | i = 1, \dots , 1/\eps, k = 1, \dots, \log_{\eps}(\delta) +1\}$ be the set of heights for vertical items. For each $h$ in $H_I$ let $w_I(h)$ be the width of all items in $i \in \mathcal{V}$ with height $h$. A configuration $C$ is a multiset of item heights out of $H_I$. We denote by $C(h)$ the number of items with height $h$ contained in $C$ and by $h(C)$ the sum of the item heights contained in $C$. 
Now consider the following linear program
\begin{align*}
\sum_{C \in \mathcal{C}} C(h)x_C & = w_I(h) & \forall h \in H_I\\
\sum_{C\in \mathcal{C}\atop h(C) \geq h} x_C & \leq \sum_{h' \in H_B, h' \geq h} w_{\mathcal{B}_V}{h'} & \forall h \in H_B
\end{align*}
$x_C$ can be interpreted as the width of the configuration $C$. With the first type of inequalities we ensure that for each type of item height we have reserved enough area to place all of these sliced items into the configurations. By the second type of inequalities we ensure that we can place all configurations into the boxes. The variable $x_C$ defines the total width of the configuration $C$. The LP has $|H_I| + |H_B| \leq 2/\eps\delta$ inequalities and at most $(|H_I|+1)^{(1+5\eps)/\delta}$ configurations. So we can find a basic solution with at most $2/\eps\delta$ non zero components in at most $\mathcal{O}((k/\eps^2)^{2/\delta}) = \mathcal{O}((1/\eps)^{24/\eps\delta})$ operations.

First, we sort the configuration by height and fill them greedily splitting the vertical jobs if necessary, such that in each configuration in each contained size there are at most two split items. Since each configuration has a height of at most $(1+5\eps)\OPT$, and the width of the vertical items is at most $\mu W$, the total area of fractional packed items is at most $2(1+5\eps)\mu/\eps\delta \cdot W\OPT$.

We put the configurations greedily into the boxes for vertical items, starting with the smallest configuration size putting it into the smallest container, which height is large enough to contain the configuration. By this packing, it can happen that some vertical items are again cut at the container borders. Since we have at most $8/(\eps^3\delta^2)$ container, the total area of the so cut items is at most $8\mu/(\eps^3\delta^2) \cdot W\OPT$. So the total area of the items jet to pack is at most $\mu(2(1+5\eps)\eps^2\delta+8)/(\eps^3\delta^2) \cdot W\OPT \leq (8+4\eps^3\delta)\mu/(\eps^3\delta^2)W\OPT$.

We are going to pack these items into the extra box $B_{V0}$. These box has a height of at most $(1/3 +\eps)\OPT$. If we rotate this box and alt the fractional items by $90$ degree, we can use the FFDH algorithm to pack these items. The generated solution has a height of at most $(6 \cdot (8+4\eps^3\delta)\mu/(\eps^3\delta^2)+ \mu) W$, since the generated packing uses at most $2 AREA(I) + (\mu W) w$ area to pack the items, where $w$ is the width of the packing area, so in this case $w > 1/3 \cdot \OPT$.  
Since $\eps \leq 1/25$ we have $(6 \cdot (18+4\eps^3\delta)\mu/(\eps^3\delta^2)+ \mu) W \leq ((6 \cdot 8+25\eps^3\delta)\mu/(\eps^3\delta^2)) W \leq \mu/2\eps^5\delta^2 \cdot W$.
Since $\mu \leq \eps^6\delta^2$ the packing has a height of at most $\frac{\epsilon}{2} W$. We now rotate the items back and this packing fits into the box of height $(1/3 + 2\eps)\OPT$ and width $\frac{\epsilon}{2} W$. Combined with the box $B_{V0}$ this is a box with width $(1-3\eps/2)W$ and height $(1/3 +\eps) \OPT$ so it matches $B_V$.


We place each configuration such that it touches the bottom of the box. It is possible that there is some free area above an inserted configuration. This free area builds a container for small items. Since we have at most $8/(\eps^3\delta^2)$ container in $\boxV$ and at most $2/\eps\delta$ configurations, there are at most $(8 + 2\eps^2\delta)/(\eps^3\delta^2) \leq 9/(\eps^3\delta^2)$ container created by this algorithm to assign the vertical items to the container. 

Since the basic solution to the linear program fulfils the equality $\sum_{C \in \mathcal{C}} C(h)x_C  = w_I(h)  \forall h \in H_I$ the configurations use exactly the area the items in $V$ used in the original packing. Therefore, the free area is at least $Area(\mathcal{B}_V\cup \{B_{V}\}) - Area(V)$.
\end{proof}

We now describe how to pack the small items. We place them into the boxes $\boxSH$ and $\boxSV$. Note that the boxes are also generated in the extra box for vertical items $B_{V0}$, which has an area of $(1-2\eps)W \times (1+\eps)\OPT$. 
The total area of the boxes for small items is at least $Area(\mathcal{B}_V\cup \{B_{V}\}) - Area(V) + AREA(\mathcal{B}_H) - AREA(H)$. Therefore, by Lemma \ref{lma:structuralLemma} it holds that $Area(\boxSH \cup\boxSV) \geq AREA(S) + (1-2\eps)(1/3+\eps)\OPT W$, since the free area in the boxes $\boxH \cup \boxV \cup \{V_0\}$ is at least $AREA(S) + (1-2\eps)(1/3+\eps)\OPT W$.
Since the vertical items do not use more space in the packing than they had before, the area we can place small items in is at least $(1-2\eps)W \times (1+\eps)\OPT$ larger than the area of the small items.

\begin{lma}
\label{lma:packingSmallItems}
There is a polynomial algorithm, that places all small items into the boxes $\boxSH$ and $\boxSV$.
The algorithm needs at most $\mathcal{O}(n\log n + 1/\eps^3\delta^2)$ operations.
\end{lma}
\begin{proof}
In $\boxSH$ and $\boxSV$ we have at most $10/\eps^3\delta^2$ boxes total.
First, we discard all boxes which are smaller than $\mu W \times \mu \OPT$ in one dimension. Each of this boxes has an area of at most $\mu \OPT W$.

We act differently for the boxes in $\boxSV$ and $\boxSH$.
The boxes in $\boxSV$ we fill with the NFDH algorithm, the boxes in $\boxSH$ and the items to be filled in are first rotated by $90$ degree and than packed with the NFDH algorithm.

Assume we can not pack all the small items into the boxes. Let us consider a box $B \in \boxSV$. Let $i$ be the first item, which was packed into $B$. We know that above the last item which is packed we have a free strip of width $w(B)$ and height at most $\mu \OPT$. Let $B_j$ be the $j$-th Strip where the NFDH algorithm places items and let $i_j$ be the first item placed in this strip. We know that the strip $B_j$ contains a free area of at most $(h(i_{j-i}) - h(i_{j}))w(B) + h(i_{j})\mu W$. The last strip $B_k$ contains a free area of at most $h(i_{k})w(B)$. So in total we have at most  $\sum_{j = 0}^{k-1} ((h(i_{j-i}) - h(i_{j}))w(B) + h(i_{j})\mu W) + h(i_{k})w(B) \leq h(i_0)w(B) +  h(B) \mu W \leq w(B) \mu \OPT + \mu \OPT W$ free space in the used strips. So the total free space in the box $B$ is at most $2w(B) \mu \OPT + \mu \OPT W \leq 3\mu \OPT W$ if we use this box and $\mu \OPT W$ if not. Analogously for each box $B \in \boxSH$ we have that the free area in each box is at most $2h(B) \mu W + \mu \OPT W \leq 3\mu \OPT W$ if it is used to fill it with items and at most $\mu \OPT W$ if not. 

So the total free area in $\boxS := \boxSV \cup \boxSH$ is at most $30\mu/\eps^3\delta^2 \cdot \OPT W$. So if $\mu \leq 1/\eps^5\delta^2$ and $\eps \leq 11$ we have $30\mu/\eps^3\delta^2 \cdot \OPT W \leq \eps \OPT W \leq \frac{1}{6}\OPT W$. But if we have at most $\frac{1}{6}\OPT W$ free area in the boxes $\boxS$, the paced small items have an area of at least $A(\mathcal{S})+(1/3+\eps)\OPT \cdot (1-2\eps)W - \frac{1}{6}\OPT W \geq A(\mathcal{S})$. Which is a contradiction to the assumption that there are some small items we could not place.
\end{proof}

To place the tall items we use the dynamic program described by Nadiradze and Wise (see Lemma 6.1 in \cite{nadiradzeWiese} ).

\begin{lma}[\cite{nadiradzeWiese}]
Given a set of bins $B_1, \dots, B_K$ with integral capacities $w(B_j)$ and a set of $n$ items, each being characterized by a size $a_i \in \mathbb{N}$. Let $\tilde{N} := \sum_{j} w(B_j)$. There is an algorithm with running time $(n \tilde{N})^{\mathcal{O}(k)}$ that determines whether there is an assignment of the $n$ items  to the $k$ bins such that each bin $B_j$ is assigned items with a total size of at most $w(B_j)$.
\end{lma}

For each item size in the tall items, we have $\mathcal{O}(1/\eps^3\delta^2)$ boxes. The total width of all tall items is at most $3W$. So for a given set of boxes for tall items we can find in $(nW)^{\mathcal{O}(1/\eps^3\delta^2)}$ operations a packing of the tall items into the boxes, or decide that such a packing does not exist.

All the described algorithms to place items have a running time of at most $(nW)^{\mathcal{O}(1/\eps^3\delta^2)}$. So we can find a placement of the items in $I$ into a given box partition with at most $(1+2\eps)/\delta^2$ boxes for large items, $(1+2\eps)/\eps\delta^2$ boxes for horizontal items, $\mathcal{O}(1/\eps^3\delta^2)$ boxes for tall items containing just items with the same height, and $\mathcal{O}(1/\eps^3\delta^2)$ boxes for vertical items in at most $(nW)^{\mathcal{O}(1/\eps^3\delta^2)}$ operations.

Let us summarize what the current packing looks like (see figure \ref{fig:partition}): We have stretched the optimal packing area, such that it has a height of $(1+5\eps)\OPT$. We have an extra box $H_0$ for horizontal items, which has height $2 \eps \OPT$ and width $W$. We place this box exactly above the packing area of height $(1+5\eps)\OPT$. For the medium sized items, we have introduced two boxes. One has height $\eps\OPT$ and can be placed above the box for horizontal items. The other has height $\tallItemHeight$ and width $(3\eps/2)W$. We will place this box next to the extra boxes for vertical items, which has height $\tallItemHeight$ and width $(1-3\eps/2)W$. 
So the total height of the current packing is $(4/3+9\eps)\OPT$. So if we substitute $\eps$ with $\eps' := \eps/9$ the simplified packing has a height of at most $(4/3 +\eps)\OPT$.


\begin{algorithm}
\caption{Given: $W$, $I$, $\eps >0$ with $1/\eps \in \mathbb{N}_{\geq 2}$}\label{euclid}
\begin{algorithmic}
\State \texttt{set} $\eps' :=\min \{\eps/9, 1/24\}$

\State\texttt{try} a value for  $\OPT$

\State\texttt{find} the corresponding values for $\delta$ and $\mu$.
\For{each position for the items in $L$}
\For{each size and position of the horizontal boxes}

\State compute the boxes $\mathcal{B}_{T \cup V}$ 

\For{each possible choice to position tall items on the box borders}

\For{ each partition of the boxes  $\mathcal{B}_{T \cup V}$ into the boxes $\boxV$ and $\boxT$ }

\State\texttt{find} for this guessed partition a packing for the items $I$ if possible.
\EndFor
\EndFor
\EndFor
\EndFor
\end{algorithmic}
\end{algorithm}

The algorithm works as follows: First we set $\eps' := \min\{\eps/9,1/24\}$. After that we have to find the height of the optimal packing $\OPT$ with a binary search framework, which takes $\mathcal{O}(\log(\OPT))$ steps. Now  we find the correct values for $\delta$ and $\mu$ and round the items in $T\cup V\cup L$. This can be done in $\mathcal{O}(n/\eps)$. 
Now we guess the structure of the packing. For this we have to guess the position of the large items ($(W/\eps\delta)^{1/\delta^2}$ possibilities), the position of the horizontal items ($W^{2/\eps\delta}$ possibilities) and the position and width of the boxes for tall items ($W^{2|\boxT|}$ possibilities). 
Since $\delta \geq \eps^{\mathcal{O}(2^{1/\eps})}$ the structure of the packing can be guessed within $W^{1/\eps^{\mathcal{O}(2^{1/\eps})}}$ operations. 
For each of the guessed partitions, we try with the algorithm from Lemma \ref{lma:algorithmToPackTheItems} if we can place the items in $I$ into that partition.
If not, we try an other partition, if yes we try a smaller value for $\OPT$.
The total running time is therefore bounded by $\mathcal{\log(\OPT)} \cdot (nW)^{1/\eps^{\mathcal{O}(2^{1/\eps})}}$. If we approximate $\OPT$ within range $(1+\eps)$ we have to scale $\eps'$ by a constant factor and get a running time of $(nW)^{1/\eps^{\mathcal{O}(2^{1/\eps})}}$.

\section{Conclusion}
We have reduced the upper bound of the approximation ratio for strip packing with pseudo polynomial processing time to $(4/3 +\eps)$. This reduced the bound by $2/30\approx 0.07$ compared to the previous best algorithm. But there is still a large gap to the lower bound of $\frac{5}{4}\OPT$. To match this lower bound no item, which has a height which is larger than $\frac{5}{4}\OPT$ is allowed to be placed outside the packing area. We believe, an algorithm with approximation ratio $(\frac{5}{4}+\eps)\OPT$ should be possible.

\section{Acknowledgements}
We kindly thank the anonymous referees for their valuable comments that helped us improve this paper. This research was supported in part by German Research Foundation (DFG) project JA 612 /14-2.

\bibliography{lib_afptas.bib}
\bibliographystyle{plain}
\newpage

\end{document}